\documentclass[11pt,a4paper,reqno,twoside]{amsart}

\usepackage{amssymb}
\usepackage{amsmath}

\usepackage[dvipsnames]{xcolor}
\usepackage{float}
\usepackage{latexsym}
\usepackage{amsthm}
\usepackage{empheq}
\usepackage{bm}
\usepackage{booktabs}
\usepackage{subcaption}
\usepackage{tikz,pgfplots,cite}
\usepackage{balance}

\usepackage{cite}
\usepackage{amsmath,amssymb,amsfonts}
\usepackage{algorithmic}
\usepackage{graphicx}

\usepackage{latexsym}
\usepackage{amsthm}
\usepackage{thm-restate}
\usepackage{empheq}
\usepackage{bm}
\usepackage{booktabs}
\usepackage{subcaption}
\usepackage{enumitem}
\usepackage{array}
\usepackage[margin=2.5cm]{geometry}

\definecolor{redish}{rgb}{0.9, 0.17, 0.31}
\definecolor{fuchs}{rgb}{0.57, 0.36, 0.51}
\usepackage[colorlinks=true,linkcolor=fuchs]{hyperref}

\theoremstyle{definition}
\newtheorem{theorem}{Theorem}
\newtheorem{corollary}{Corollary}
\newtheorem{proposition}{Proposition}
\newtheorem{conjecture}{Conjecture}
\newtheorem{definition}{Definition}
\newtheorem{example}{Example}
\newtheorem{notation}{Notation}
\newtheorem{remark}{Remark}

\newtheorem{lemma}{Lemma}

\newtheorem{prob}{Problem}

\usepackage{comment}

\newcommand{\F}{\mathbb{F}}

\newcommand{\E}{\mathbb{E}}

\newcommand{\mC}{\mathcal{C}}
\newcommand{\mS}{\mathcal{S}}

\newcommand{\mA}{\mathcal{A}}
\newcommand{\mB}{\mathcal{B}}

\newcommand{\mR}{\mathcal{R}}
\newcommand{\mG}{\mathcal{G}}
\newcommand{\mH}{\mathcal{H}}
\newcommand{\mP}{\mathcal{P}}

\newcommand{\supp}{\operatorname{supp}}

\newcommand{\tmax}{T_{\max}}

\newcommand\qbin[3]{\left[\begin{matrix} #1 \\ #2 \end{matrix}\right]_{#3}}

\newcommand\binomi[2]{\left(\begin{matrix} #1 \\ #2 \end{matrix} \right)}

\newtheorem{claim}{Claim}

\newcommand\rk{\textnormal{rk}}

\def\BibTeX{{\rm B\kern-.05em{\sc i\kern-.025em b}\kern-.08em
    T\kern-.1667em\lower.7ex\hbox{E}\kern-.125emX}}

\newcommand{\pcomment}[1]{{\color{purple}#1}}
\newcommand{\vcomment}[1]{{\color{violet}#1}}

\newcommand{\ey}[1]{{{\footnotesize [\pcomment{#1}\;\;\vcomment{--Eitan}]}}}

\title[A Combinatorial Perspective on Random Access Efficiency for DNA Storage]{\textbf{A Combinatorial Perspective on \\ Random Access Efficiency for DNA Storage}}
\usepackage[foot]{amsaddr}
\author{Anina Gruica$^1$}
\address{$^1$Technical University of Denmark, Lyngby, Denmark.}
\thanks{$^1$A. G. is supported by the Dutch Research Council through grant OCENW.KLEIN.539 and by the Villum Fonden through grant VIL”52303”.}
\email{$^1$anigr@dtu.dk}

\author{Daniella Bar-Lev$^2$}
\address{$^2$Technion -- Israel Institute of Technology, Haifa, Israel.}
\thanks{$^2$The research of D. B. and E. Y. is funded by the European Union (ERC, DNAStorage, 101045114). Views and opinions expressed are however those of the authors only and do not necessarily reflect those of the European Union or the European Research Council Executive Agency. Neither the European Union nor the granting authority can be held responsible for them. The work of D. B. and E. Y. is also funded in part by NSF under Grant CCF2212437.}
\email{$^2$ \{daniellalev,yaakobi\}@cs.technion.ac.il}
\author{Alberto Ravagnani$^3$}
\address{$^3$Eindhoven University of Technology, Eindhoven, the Netherlands.}
\email{$^3$a.ravagnani@tue.nl}
\thanks{$^3$A. R. is supported by the Dutch Research Council through grants VI.Vidi.203.045 and OCENW.KLEIN.539.}

\author{Eitan Yaakobi$^2$}

\thanks{Part of this work was presented and published at the IEEE International Symposium on Information Theory (ISIT), Athens, Greece, 2024~\cite{gruica2024reducing}.}

\usepackage{setspace}
\begin{document}
\maketitle

\begin{abstract}
We investigate the fundamental limits of the recently proposed \emph{random access coverage depth problem} for DNA data storage. Under this paradigm, it is assumed that the user information consists of $k$ information strands, which are encoded into $n$ strands via a generator matrix $G$. During the sequencing process, the strands are read uniformly at random, as each strand is available in a large number of copies. In this context, the random access coverage depth problem refers to the expected number of reads (i.e., sequenced strands) required to decode a specific information strand requested by the user. This problem heavily depends on the generator matrix $G$, and besides computing the expectation for different choices of $G$, the goal is to construct matrices that minimize the maximum expectation over all possible requested information strands, denoted by $\tmax(G)$.

In this paper, we introduce new techniques to investigate the random access coverage depth problem, capturing its combinatorial nature and identifying the structural properties of generator matrices that are advantageous. We establish two general formulas to determine $\tmax(G)$ for arbitrary generator matrices. The first formula depends on the linear dependencies between columns of $G$, whereas the second formula takes into account recovery sets and their intersection structure. We also introduce the concept of \emph{recovery balanced codes} and provide three sufficient conditions for a code to be recovery balanced. These conditions can be used to compute $\tmax(G)$ for various families of codes, such as MDS, simplex, Hamming, and binary Reed-Muller codes. Additionally, we study the performance of modified systematic MDS and simplex matrices, showing that the best results for $\tmax(G)$ are achieved with a specific combination of encoded strands and replication of the information strands.

\end{abstract}

\section{Introduction}
With the demand for storage capacity consistently outpacing the capabilities of existing technologies to store it~\cite{rydning2022worldwide}, there is a critical need for alternative approaches. In response to this pressing challenge, DNA storage emerges as a promising solution for long-term data storage, offering exceptional density and durability \cite{alliance2021preserving, markowitz2023biology}. A typical DNA storage system is composed of three main components: DNA synthesis, storage containers, and DNA sequencing. Initially, synthetic DNA strands, or \emph{oligos}, are created to encode the user's information. These strands are then stored in an unordered manner in a storage container. Subsequently, DNA sequencing translates the stored strands into digital sequences, called \emph{reads}, that should be decoded back to the user's information. Due to current technology limitations, this process results in multiple noisy copies for each designed strand, that are obtained without order.

While several works demonstrated the potential of DNA as a storage medium \cite{anavy2019data,blawat2016forward,bornholt2016dna, organick2018random, yazdi2017portable, tabatabaei2015rewritable, bar2021deep}, the efficiency of DNA sequencers remains a bottleneck, with slow throughput and high costs compared to alternative storage technologies \cite{shomorony2022information, yazdi2015dna, alliance2021preserving}. This bottleneck is intricately tied to the concept of the \emph{coverage depth} \cite{heckel2019characterization}, defined as the ratio between the number of sequenced reads and the number of designed DNA strands. Reducing the coverage depth presents an opportunity for a  drastic improvement in latency and cost reduction~\cite{erlich2017dna, chandak2019improved, bar2023cover, preuss2024sequencing}. 

Our point of departure in this paper is the recent work~\cite{bar2023cover} which initiates the study of the so-called \emph{DNA coverage depth problem}. The aim of this problem is to reduce the cost and latency of DNA sequencing, by studying the expected number of reads that are required in order to retrieve the user's information. The authors of~\cite{bar2023cover} investigated both the non-random and random access scenarios. While in the former the goal is to retrieve \emph{all} the user's information, the latter considers the random access case of retrieving only a \emph{single} strand. Assume the~$k$ information strands (representing the data) are encoded into $n$ strands. By drawing a connection to the coupon collector's problem~\cite{erdHos1961classical, felleb1968introduction, flajolet1992birthday, newman1960double}, if the $k$ information strands are encoded by an MDS code, then the expected number of reads to decode all $k$ information strands is $H_n-H_{n-k}$, where $H_i$ is the $i$-th harmonic number.
This result is optimal for minimizing the expected number of reads. Using the same MDS code for the random access case, it was proven that the expected number of reads is $k$, which surprisingly can also be achieved by not applying any code to the information strands. Moreover, this result is far from optimality as in~\cite{bar2023cover} it was shown that codes with expectation $c \cdot k$, for $c < 1$, exist. 

A related concept was also explored in \cite{chandak2019improved}, where the authors investigated the trade-offs between the reading costs, which are directly associated with coverage depth, and the writing costs. The non-random access scenario of the DNA coverage depth problem was further extended in~\cite{cohen2024optimizing} to support the setup of composite DNA letters~\cite{anavy2019data}, and in~\cite{preuss2024sequencing, sokolovskii2024coding} for the setup of the combinatorial composite of DNA shortmers~\cite{preuss2021efficient}.
Another extension to the random access setup was studied in~\cite{AGY24};  however, the goal was  not to decode a single strand but rather a group of strands that constitute a single file.

Motivated by the results and observations in \cite{bar2023cover}, in this work, we focus on the random access coverage depth problem for linear codes. The rest of the paper is organized as follows. In Section~\ref{sec:def}, we formally define the random access coverage depth problem. Section~\ref{sec:general} presents several properties of the random access expectation and, more precisely, gives two general formulas for the expectation. 
In Section~\ref{sec:average}, we show an important observation about the average of the random access expectation, which is then used in Section~\ref{sec:rec_bal} to study in more detail codes that have a very balanced behavior in terms of the random access problem. We call these codes \emph{recovery balanced codes}, and we show that for these codes the expectation is always equal to $k$. Through applying three sufficient conditions, we demonstrate that certain families of codes (MDS, Hamming, simplex, binary Reed-Muller, binary Golay) have random access expectation~$k$. 
Furthermore, in Section~\ref{sec:rec_bal}, we discuss code operations that preserve, or do not preserve, the property of being recovery balanced. In particular, we show that if the \emph{permutation automorphism group} of a code~\cite{macwilliams1977theory} is transitive, then the property of being recovery balanced is preserved under duality, and we conjecture that, without assuming any conditions, this property always holds true. 
From the results in Section~\ref{sec:rec_bal}, it is evident that codes that are recovery balanced are not good candidates for the random access problem, since they have the same random access expectation as the uncoded case (i.e., the case where the information strands and the encoded strands are the same). Motivated by the latter, in Section~\ref{sec:break}, we demonstrate that ``breaking'' the balance of recovery balanced codes can reduce the random access expectation strictly below~$k$. More specifically, this presents a method to derive generator matrices for which the random access expectation is smaller than~$k$. An analysis of this method, as well as several observations and experimental results, are given in Section~\ref{sec:break}. Finally, Section~\ref{sec:conclusion} concludes the paper and proposes some open questions for future research.


\section{Problem Statement}\label{sec:def}

Throughout this paper, $k$ and $n$ are positive integers with $2 \le k \le n$, 
$q$ denotes a prime power, and $\F_q$ is the finite field with $q$ elements.
For a positive integer $n$, we let
$[n]=\{1, \ldots,n\}$ and denote by $H_n$ the $n$-th harmonic number, i.e., $H_n:=1+1/2+\dots+1/n$.

We study the expected sample size for uniformly random access queries in DNA storage systems. In DNA-based storage systems, the data is stored as a length-$k$ vector of sequences (called \textit{strands}) of length $\ell$ over the alphabet $\Sigma=\{A,C,G,T\}$. We embed $\Sigma^\ell$ into a finite field $\F_q$ and use
a $k$-dimensional linear block
code $\mC \subseteq \F_q^n$
to encode an information vector
$\smash{\bm{x} = (x_1,\dots,x_k) \in (\Sigma^\ell)^k} \subseteq \F_q^k$ to an encoded vector $\smash{\bm{y} = (y_1,\dots,y_n) \in \F_q^n}$. Note that in order to embed $\Sigma^\ell$ into a finite field $\F_q$, we would need $|\Sigma^{\ell}|=4^{\ell}$ to divide $q$, however, we consider any prime power $q$ in this paper, without any restrictions on $q$.

To retrieve the stored information at a later time, the strands are amplified and then sequenced using DNA sequencing technology. This generates multiple (erroneous) copies for different strands, referred to as \emph{reads}. To simplify our analysis, in this paper, we will assume that this step is accomplished error-free. The output of the reading process is a multiset of these reads, without any specific order\footnote{The reads can be obtained all together or one after the other, depending on the specific technology that is being used.}. As current prices and throughput for DNA sequencing still lag behind other archival storage solutions, reducing the coverage depth required for information recovery is crucial.

In the random access setup, the goal is to retrieve a single information strand $x_i$ for $i \in [k]$. It has been demonstrated in~\cite{bar2023cover} that the expected sample size of a random access query in the DNA storage system can be decreased using an appropriate coding scheme. We illustrate this concept with the following example.


\begin{example}\label{ex:introE}
We wish to store an information vector of size $k=2$, namely $(x_1,x_2) \in \F_q^2$. Without coding, the expected number of samples that are needed to recover each of the two information strands is 2 (assuming 
that the samples are chosen uniformly at random).
If $\F_2 \subseteq \F_q$, we can consider the matrix $$G=\begin{pmatrix} 
    1 & 0 & 1 & 0 & 1  \\
    0 & 1 & 0 & 1 & 1\\
\end{pmatrix} \in \F_2^{2 \times 5} $$ 
and store the entries of
$$(x_1,x_2)G=(x_1,x_2,x_1,x_2,x_1+x_2) \in \F_q^5. $$
This time, using uniformly random sampling of the five encoded symbols, it can be shown that the expected number of samples that are needed to recover either of the two information strands is approximately $1.917<2$. Note that ``recovering'' means being able to obtain the original information strand as a linear combination of the sampled symbols. For example, if the last two encoded strands are sampled, the information strand $x_1$ can be recovered as
$x_1=x_2+(x_1+x_2)$.
\end{example}

As demonstrated in Example~\ref{ex:introE}, once the $k$ information strands are encoded using a generator matrix $G \in \F_q^{k \times n}$ it is possible to refer to every read of an encoded strand as reading its corresponding column in the matrix $G$, and recovering the~{$i$-th}~information strand corresponds to recovering the $i$-th basis vector, that is, it should belong to the $\F_q$-span of the already recovered columns of $G$. Motivated by these observations, we are now ready to formally define the problem studied in this paper. We note that in~\cite{bar2023cover}, this problem was referred to as the \emph{singleton coverage depth problem} (Problem 3) and here we refer to it as the \emph{random access coverage depth problem}.

\begin{prob}[{\textbf{The random access coverage depth problem}}] \label{prob:1}
Let $G \in \F_q^{k \times n}$ be a rank-$k$ matrix. Suppose that the columns of~$G$ are drawn uniformly at random, meaning that each column has probability $1/n$ of being drawn and columns can be drawn multiple times. For $i\in [k]$, let $\tau_i(G)$ denote the random variable that counts the minimum number of columns of $G$ that are drawn until the standard basis vector~$e_i$ is in their $\F_q$-span. Compute the expectation $\E[\tau_i(G)]$ and the maximum expectation $$\tmax(G) \triangleq \max_{i\in[k]}\E[\tau_i(G)].$$ Furthermore, let \smash{$T_q(n,k) \triangleq \min_{G \in \F_q^{k \times n}}{\tmax(G)}$} be the smallest possible maximum random access expectation over all rank-$k$ matrices in $\F_q^{k \times n}$ and $T_q(k) \triangleq \liminf_{n\rightarrow\infty}{T_q(n,k)}$ be the best maximum random access for any rank-$k$ matrix over $\F_q$.
\end{prob}

Note that for Problem~\ref{prob:1}, we are only concerned with the multiset of vectors made from the columns of the matrix~$G$, and so the order of these columns is irrelevant.

\begin{remark}
    In the sequel, $G \in \F_q^{k \times n}$ denotes a rank-$k$ matrix and $\mC \subseteq \F_q^n$ the $k$-dimensional code having $G$ as its generator matrix, i.e., $\mC$ is the $\F_q$-span of the rows of $G$. We do not assume that $G$ is \textbf{systematic} (i.e., that the first $k$ columns of $G$ form the identity $k \times k$ matrix), unless otherwise specified.
    Note that, in contrast with previous approaches, we mostly focus on generator matrices and not on block codes. This is because the parameters we consider in this paper depend on the choice of the generator matrices and not only on the code.  
\end{remark}

Studying the values $\tmax(G)$, $T_q(n,k)$, and $T_q(k)$ was initiated in~\cite{bar2023cover}. It was established that for various codes, such as the code generated by the identity matrix, the simple parity code, and MDS codes, 
for any $i\in[k]$ we have that $\E[\tau_i(G)] = k$, when $G$ is the systematic generator matrix for any of these codes. In particular, the result on the identity codes implied that $T_q(k,k)=k$, but finding in general the value of $T_q(n,k)$ is an intriguing question, and first steps towards solving this value were carried in~\cite{bar2023cover}. Several constructions were presented that achieve maximum expectation strictly lower than~$k$. More specifically, it was show that $T_q(k=2)\leq 0.91\cdot 2$, $T_q(k=3)\leq 0.89\cdot 3$ and for arbitrary $k$ which is a multiple of $4$, it holds that $T_q(n=2k,k) \leq 0.95k$. Furthermore, two lower bounds from~\cite{bar2023cover} established that for any $n,k,q$ it holds $T_q(n,k) \geq n - \frac{n(n-k)}{k} (H_n-H_{n-k})$ and $T_q(k) \geq \frac{k+1}{2}$.


Despite these valuable contributions, the fundamental limits of the random access coverage depth problem remain unclear. 
Specifically, existing results lack a comprehensive understanding of the properties that render a generator matrix optimal for this purpose, as well as how to calculate the value $\tmax(G)$ in general. In this context, our goal is to contribute to the ongoing research and identify generator matrices that minimize $\tmax(G)$ and help in determining the values of $T_q(n,k)$ and~$T_q(k)$.

\section{General Formulas for Expectation}\label{sec:general}

In this section, given any matrix $G\in \F_q^{k \times n}$, we give two general formulas for the expected number of reads until one can recover the $i$-th information strand, namely, for $\E[\tau_i(G)]$. This will allow us to better understand which properties of 
$G$ play a role in the solution of Problem~\ref{prob:1}. 
As we will see in Section~\ref{sec:rec_bal}, the choice between the two formulas depends on the class of codes under investigation. While the first formula may be more intuitive for certain codes, the second formula appears to be more suitable for others. 

\begin{definition}
We call $S \subseteq [n]$ a \textbf{recovery set} for the $i$-th information strand if $e_i$ is in the span of the columns of $G$ indexed by $S$. We denote by $\mR(i)$ the set of minimal (with respect to inclusion) recovery sets for the $i$-th information strand. We say that ``we recovered the $i$-th strand'' if we drew columns of $G$ whose indices form a recovery set for~$i$.
\end{definition}

By definition, the columns of $G$ are in one-to-one correspondence with the encoded strands. Denote the $j$-th column of~$G$ by $g_j$ for $j \in [n]$ and for~$i\in [k]$ and $0 \le s \le n$, let 
\begin{align*}
\alpha_i(s) := {|\{S \subseteq [n] : |S| = s, \, e_i \in \langle g_j : j \in S\rangle \}|}.
\end{align*}
The first formula we establish for $\E[\tau_i(G)]$ uses the values~$\alpha_i(s)$ we just introduced.

\begin{lemma} \label{lem:fi}
For $G\in \F_q^{k \times n}$ and for all $i\in[k]$ we have
\begin{align*}
    \E[\tau_i(G)] =  n H_n - \sum_{s=1}^{n-1} \displaystyle \frac{\alpha_i(s)}{\binom{n-1}{s}}.
\end{align*}
\end{lemma}
\begin{proof}
By definition,
\begin{align} \label{eq:simp1}
    \E[\tau_i(G)] = \sum_{r=1}^\infty \Pr[\tau_i(G) \ge r].
\end{align}
For $r \ge 1$, let $\eta_{r}$ be the random variable that denotes the number of distinct encoded strands that were sampled in the first $r$ draws, where we set $\eta_0=0$. We have
\begin{equation} \label{eq:simp2}
\Pr[\tau_i(G) \ge r] = 
\sum_{s=0}^{n-1}\Pr[\tau_i(G) \ge r \mid \eta_{r-1} = s] \Pr[\eta_{r-1} = s],
\end{equation}
where the sum runs only up to $n-1$ since $\Pr[\tau_i(G) \ge r \mid \eta_{r-1} = n] = 0$ for all $r \ge 1$.
In order to compute the probability $\Pr[\tau_i(G) \ge r \mid \eta_{r-1} = s]$ we note that the number of subsets of $s$ strands that recover the $i$-th strand is the same as the number of subsets $S\subseteq [n]$ with $|S| = s$ that have the property that the columns of $G$ indexed by $S$ contain $e_i$ in their span, which is exactly~$\alpha_i(s)$.
Therefore, the number of sets of cardinality $s$ that do not recover the $i$-th information strand is $\binom{n}{s}-\alpha_i(s)$.
Combining this with~\eqref{eq:simp1} and~\eqref{eq:simp2} gives
\begin{align*}
    \E[\tau_i(G)] = \sum_{r=1}^\infty \sum_{s=0}^{n-1}\left(1-\frac{\alpha_i(s)}{\binom{n}{s}}\right)  \Pr[\eta_{r-1} = s].
\end{align*}
Using the Inclusion-Exclusion Principle, we further obtain
\begin{align*}
    \Pr[\eta_{r-1} = s] =\binom{n}{s}\displaystyle\sum_{j=0}^{s}\binom{s}{j}(-1)^j\left(\frac{s-j}{n}\right)^{r-1}.
\end{align*}
This gives 
\allowdisplaybreaks
\begin{align*}
    \E[\tau_i(G)]&= \sum_{s=0}^{n-1}\left(\binom{n}{s}-{\alpha_i(s)}\right) \displaystyle\sum_{j=0}^{s}\binom{s}{j}(-1)^j\sum_{r=0}^\infty\left(\frac{s-j}{n}\right)^{r} \\
    &= \sum_{s=0}^{n-1}\left(\binom{n}{s}-{\alpha_i(s)}\right) \displaystyle\sum_{j=0}^{s}\binom{s}{j}(-1)^j\frac{n}{n-s+j} \\ 
    &= \sum_{s=0}^{n-1}\left(\binom{n}{s}-{\alpha_i(s)}\right) \displaystyle \frac{1}{\binom{n-1}{s}} = n  H_n - \sum_{s=0}^{n-1} \displaystyle \frac{\alpha_i(s)}{\binom{n-1}{s}},
\end{align*}
where in the second-to-last equality we used the identity
\begin{align} \label{eq:dani}
    \sum_{j=0}^{s}\binom{s}{j}(-1)^j\frac{n}{n-s+j} = \frac{1}{\binom{n-1}{s}} \quad \mbox{for $0 \le s\le n-1$},
\end{align}
which can be shown by induction. Finally, since $\alpha_i(0)=0$ for all $i \in [n]$ we obtain the statement in the lemma.
\end{proof}

\begin{remark}
If, instead of wanting to recover a single information strand, we want to recover a subset of information strands, Lemma~\ref{lem:fi} can easily be adjusted to this case. More precisely, say we want to recover all the information strands indexed by some set $I = \{i_1, \dots, i_{\ell}\} \subseteq [k]$. We then define  
\begin{align*}
    \alpha_I(s)=|\{S \subseteq [n] : |S| = s, \langle e_i : i \in I\rangle \subseteq \langle g_j : j \in S\rangle \}|.
\end{align*}
Following reasoning analogous to the one used in the proof of Lemma~\ref{lem:fi}, we can conclude that the expected number of encoded strands that need to be drawn in order to recover the information strands indexed by $I$ is
\begin{align*}
   nH_n-\sum_{s=1}^{n-1} \frac{\alpha_I(s)}{\binom{n-1}{s}}. 
\end{align*}
Note that the generalization of Lemma~\ref{lem:fi} discussed in this remark was originally observed and brought to our attention by M. Bertuzzo in~\cite{matteo}.
\end{remark}

We illustrate 
how Lemma~\ref{lem:fi} can be used to compute $\E[\tau_1(G)]$ and $\E[\tau_2(G)]$ for the matrix $G$ in Example~\ref{ex:introE}.

\begin{example}
Let $G$ be as in Example~\ref{ex:introE}.
We have
\begin{align*}
\alpha_1(1) = 2, \; \alpha_1(2) = 9, \; \alpha_1(3) = \binom{5}{3}, \; \alpha_1(4) = \binom{5}{4}.
\end{align*}
By Lemma~\ref{lem:fi} we obtain 
\begin{align*}
    \E[\tau_1(G)] &= 5 H_5 - \sum_{s=1}^{4} \displaystyle \frac{\alpha_1(s)}{\binom{4}{s}} 
    = \frac{23}{12} \approx 1.917.
\end{align*}
It is easy to see that for $\tau_2(G)$ we have exactly the same numbers, and thus $\E[\tau_2(G)] = 23/12$ and $\tmax(G) = 23/12.$
\end{example}

Using Lemma~\ref{lem:fi} we can also give an alternative (and shorter) proof of~\cite[Theorem 9]{bar2023cover}.

\begin{corollary} \label{cor:mds}
    Let $G\in \F_q^{k \times n}$ be a systematic generator matrix of an MDS code. We have $\tmax(G)=k$.
\end{corollary}
\begin{proof}
Since $G$ is an MDS matrix, every $k$ columns of $G$ are linearly independent. Thus, for any~$i \in [k]$,
\begin{align*}
   {\alpha}_i(s) = \begin{cases}
      \binom{n-1}{s-1} \quad &\textnormal{if $s \in [k-1]$,} \\
       \binom{n}{s} \quad &\textnormal{if $s \ge k$.}
    \end{cases}
\end{align*}
By Lemma~\ref{lem:fi} we then have
\begin{align*}
    \E[\tau_i(G)] &= n H_n - \sum_{s=1}^{k-1}\frac{\binom{n-1}{s-1}}{\binom{n-1}{s}}-\sum_{s=k}^{n-1}\frac{\binom{n}{s}}{\binom{n-1}{s}} \\
    &= n  H_n - \sum_{s=1}^{k-1}\frac{s}{n-s}-\sum_{s=k}^{n-1}\frac{n}{n-s}, 
    \end{align*}
which simplifies to $k$ after straightforward computations.
\end{proof}

The other formula for the expectation of the random access problem we will use in this paper relies on the {recovery sets} of the information strand and on their intersection structure. 
%
%
%
%
%
More precisely, for $i \in [k]$ and $\mR(i)=\{R_1,\dots,R_L\}$ we denote by $\beta_i(s,j)$ the number of subsets
$S \subseteq [L]$ of cardinality $s$ such that $\bigcup_{h \in S} R_h$ has cardinality $j$. In symbols, $\beta_i(s,j) = \left| \{ S \subseteq [L] : |S|=s, \,|\bigcup_{h \in S} R_h | = j  \}  \right|$. We then have the following result, which can be obtained analogously to~\cite[Theorem 8, Corollary 3]{bar2023cover}. We include the proof for completeness.

\begin{lemma} \label{lem:bsj}
For $G\in \F_q^{k \times n}$ and for $i \in [k]$ it holds that
\begin{align*}
    \E[\tau_i(G)]=n \Big( \sum_{j=1}^{n} H_j\sum_{s=1}^L (-1)^{s+1} \beta_i(s,j) \Big).
\end{align*}
\end{lemma}
\begin{proof}
We only prove the result in detail for the case where $\mR(i) = \{A, B\}$.
First, note that we can represent any sequence of $r-1$ draws as an $(r-1)$-tuple $\mathbf{d} \in [n]^{r-1}$, 
where $d_j$ denotes the column number in $G$ that was obtained in the $j$-th draw.
For $\mathbf{d} = (d_1, \dots, d_{r-1}) \in [n]^{r-1}$, let $\varphi := \{d_i : 1 \le i \le r-1\}$. For a set of indices $S \subseteq [n]$, let $\lambda_{S}(r-1)$ be the number of different ways of drawing columns in the first $r-1$ draws such that, for at least one of the indices $j \in S$, the $j$-th column of $G$ was not drawn. We have
\begin{align*}
   \lambda_{A \cup B}(r-1)  &=  |\{\mathbf{d} \in [n]^{r-1} : \varphi(\mathbf{d}) \not\supseteq A \cup B\}|  \\
   &= |\{\mathbf{d} \in [n]^{r-1} : \varphi(\mathbf{d}) \not\supseteq A \} \cup \{\mathbf{d} \in [n]^{r-1} : \varphi(\mathbf{d}) \not\supseteq B\}|.
\end{align*}
Moreover, let $\lambda(r-1)$ be the number of different ways of drawing columns in the first $r-1$ draws such that the $i$-th basis vector is not in their span. Note that since $\mR(i) = \{A, B\}$, we have that $\lambda(r-1)$ is the number of different ways of drawing columns in the first $r-1$ draws such that at least one column indexed by an element in $A$ and at least one column indexed by an element in $B$ were not drawn. Thus, we have
\begin{align*}
   \lambda_{A \cup B}(r-1) 
   &= |\{\mathbf{d} \in [n]^{r-1} : \varphi(\mathbf{d}) \not\supseteq A \}| + |\{\mathbf{d} \in [n]^{r-1} : \varphi(\mathbf{d}) \not\supseteq B\}| \\
   &\qquad - |\{\mathbf{d} \in [n]^{r-1} : \varphi(\mathbf{d}) \not\supseteq A \} \cap \{\mathbf{d} \in [n]^{r-1} : \varphi(\mathbf{d}) \not\supseteq B\}| \\
   &= \lambda_{A}(r-1) + \lambda_{B}(r-1) - \lambda(r-1).
\end{align*}
The conclusion of the proof can then be completed in exactly the same way as in the proof of~\cite[Theorem 8]{bar2023cover}.

For the more general case of having more than just two recovery sets, we have
\begin{align*}
   \E[\tau_i(G)] &= n \Bigg( \sum_{s=1}^L (-1)^{s+1} \sum_{1 \le j_1 < \dots < j_s \le L} H_{|R_{j_1} \cup \dots \cup R_{j_s}|} \Bigg) \\
   &= n \Bigg( \sum_{s=1}^L (-1)^{s+1} \sum_{j=1}^{n} \beta_i(s,j) H_j \Bigg) \\
   &= n \Bigg( \sum_{j=1}^{n} H_j \sum_{s=1}^L (-1)^{s+1} \beta_i(s,j) \Bigg),
\end{align*}
where the first equality can be obtained analogously to the case of only two recovery sets by using the Inclusion-Exclusion Principle.
\end{proof}

Although Lemmas~\ref{lem:fi} and~\ref{lem:bsj} provide general formulas for the expectation $\E[\tau_i(G)]$, they are not necessarily easy to apply. This is because for general codes it is not immediate to derive the values of $\alpha_i(s)$ or $\beta_i(s,j)$, and even if one has their explicit values, it can be cumbersome to evaluate the formulas for $\E[\tau_i(G)]$. In particular, although the formulas give insight about what matters for computing $\E[\tau_i(G)]$, it is still unclear how to use them to find matrices/codes for which the random access expectation value is below $k$. Yet, in the next section we show how to build on these results in order to develop deeper insights into how one can obtain codes for the purpose of expectation below $k$.

\section{The Average Random Access Expectation}\label{sec:average}

In this section we present one of the main results of this paper. We obtain this result by investigating a random variable that, while different, is closely related to $\tau_i(G)$. This new random variable exhibits useful properties that help us better understand which structural characteristics of generator matrices (and codes) are desirable for achieving random access expectation below~$k$.

\begin{notation}
For $i \in [n]$, we denote by $\widetilde{\tau}_i(G)$ the random variable that counts the number of columns of $G$ that need to be drawn until the $i$-th column of $G$ can be recovered, i.e., until the $i$-th column belongs to the $\F_q$-span of drawn columns.
\end{notation}

Note that by definition if the matrix $G$ is systematic, then $\widetilde{\tau}_i(G) = {\tau}_i(G)$ for all~$i \in [k]$. A simple and well-known observation on minimal recovery sets is stated in the next claim.

\begin{claim} \label{claim:dual}
Let $\mC \subseteq \F_q^n$ be a code of dimension $k$
with systematic generator matrix $G$ and let $i \in [k]$. There exists a (minimal) recovery set $R \subseteq [n]$ with $|R|\ge2$ for the $i$-th encoded strand if and only if there is a codeword $x \in \mC^\perp$ with $i \in \supp(x)$ and $R=\supp(x) \setminus \{i\}$.
\end{claim}

It follows from Claim~\ref{claim:dual} that in contrast to $\tau_i(G)$, $\widetilde{\tau}_i(G)$ only depends on the code
generated by $G$. We prove this claim formally.

\begin{claim}\label{claim_good}
For a code $\mC$ with generator matrices $G,G' \in \F_q^{k \times n}$, we have $\E[\widetilde{\tau}_i(G)]=\E[\widetilde{\tau}_i(G')]$ for all~$i \in [n]$.
\end{claim}
\begin{proof}
Denote the $\ell$-th column of $G$ by $g_{\ell}$. Similarly to the proof of Lemma~\ref{lem:fi} we can show that
\begin{align} \label{eq:aaa}
    \E[\widetilde{\tau}_i(G)] =  n H_n - \sum_{s=1}^{n-1} \displaystyle \frac{\widetilde{\alpha}_i(s)}{\binom{n-1}{s}}
\end{align}
where $\widetilde{\alpha}_i(s) := {|\{S \subseteq [n] : |S| = s, \, g_i \in \langle g_{\ell} : \ell \in S\rangle \}|}.$
We have
    \begin{align*}
    \widetilde{\alpha}_i(s)&={|\{S \subseteq [n] : |S| = s, \, g_i \in \langle g_{\ell} : \ell \in S\rangle \}|} \\
    &= {|\{S \subseteq [n]\setminus \{ i \} : |S| = s, \exists x \in \mC^\perp : i \in \supp(x) \subseteq S \cup \{i\}\}|} \\&+ {|\{S \subseteq [n] : |S| = s, \, i \in S\}|}  \\
    &= {|\{S \subseteq [n]\setminus \{ i \} : |S| = s, \exists x \in \mC^\perp : i \in \supp(x) \subseteq S \cup \{i\}\}|} + \binom{n-1}{s-1},
\end{align*}
which solely depends on the code $\mC$.
Therefore, since by~\eqref{eq:aaa} we have that $\E[\widetilde{\tau}_i(G)]$ only depends on $\widetilde{\alpha}_i(s)$, we obtain the statement of the claim.
\end{proof}

The following result shows that, \emph{on average}, for any rank~$k$ matrix $G \in \F_q^{k \times n}$ the expectation of $\widetilde{\tau}_i(G)$ is $k$.

\begin{theorem} \label{thm:sumra}
Let $\mC \subseteq \F_q^n$ be a code of dimension $k$ with generator matrix $G$. We have 
$\sum_{i=1}^n\E\left[ \widetilde{\tau}_i(G)\right] = kn.$
\end{theorem}
To provide intuition for the proof of Theorem~\ref{thm:sumra}, we start with an example. We introduce some new random variables that we will use in the proof: For a generator matrix $ G \in \mathbb{F}_q^{k \times n} $ and $ i \in [k] $, let $ t_i(G) $ be the random variable that represents the number of draws until we read the $ i $-th linearly independent column, after having read $ i-1 $ linearly independent columns. We also define $ \delta_i(G) $ to be the random variable representing the number of columns that are recovered as linear combinations of the first $ i $ linearly independent columns that were read (in addition to what was already recovered with the first $ i-1 $ linearly independent columns). Furthermore, we define the random variable $ T_i(G) = \sum_{\ell=1}^i t_{\ell}(G) $, which represents the number of draws until we have recovered $ i $ linearly independent columns. We observe that
\begin{align} \label{eq:rewrite}
    \sum_{i=1}^n \widetilde{\tau}_i(G) = \sum_{i=1}^k \delta_i(G) T_i(G) = \sum_{i=1}^k \delta_i(G) \sum_{\ell=1}^i t_\ell(G).
\end{align}

\begin{example} \label{ex:prof}
Let $\mC$ be the $q$-ary simplex code of dimension~$4$ over~$\F_2$ and let $G\in \F_2^{4 \times 15}$ be its generator matrix. 
Note that since the columns of $G$ are all the non-zero vectors in $\F_2^{4}$, every time we read a new (linearly independent) column, we recover not only this column, but also the sum of this column with any of the previously recovered columns. Therefore, $\delta_i(G)$ is deterministic in this example for all $i \in [4]$ and we have
\begin{align} \label{eq:ham1}
\delta_1(G) = 1, \, \delta_2(G) = 2,  \, \delta_3(G) = 4, \, \delta_4(G) = 8.
\end{align}
Moreover, since $t_i(G)$ is a geometric random variable with success probability ${\left(15-\delta\right)}/15$, where $\delta = \sum_{\ell=1}^{i-1} \delta_{\ell}(G) $, we have
\begin{align} \label{eq:ham2}
&\E[t_1(G)] = {15}/{15}, \; \E[t_2(G)] = {15}/{14}, \nonumber \\
&\E[t_3(G)] = {15}/{12}, \; \E[t_4(G)] = {15}/{8}.
\end{align}
Combining~\eqref{eq:rewrite} with~\eqref{eq:ham1} and~\eqref{eq:ham2}, we obtain
\begin{align*}
    \sum_{i=1}^{15} \E\left[ \widetilde{\tau}_i(G) \right] &= \sum_{i=1}^4\delta_i(G)\sum_{\ell=1}^i \E\left[t_\ell(G) \right]= 4\cdot 15 = 60.
\end{align*}
\end{example}

We can now prove Theorem~\ref{thm:sumra} using the random variables used before Example~\ref{ex:prof} following analogous reasoning.

\begin{proof}[Proof of Theorem~\ref{thm:sumra}]
For $\ell \in [k]$, we let $\Delta_{\ell}:=\sum_{i=1}^{\ell-1} \delta_i(G)$. Notice that if $\Delta_{\ell}=\delta$, then $t_i(G)$ is a geometric random variable with success probability ${\left(n-\delta\right)}/n$  and so
\begin{align} \label{eq:condprob}
    \E\Big[t_i(G) \, \big| \, \Delta_{\ell}=\delta \Big]= \frac{n}{n-\delta}\; \textnormal{ for all $i \in [k]$}.
\end{align}
Rewriting~\eqref{eq:rewrite} we have
\begin{align*}
    \sum_{i=1}^n \widetilde{\tau}_i(G) =  \sum_{i=1}^k\delta_i(G)\sum_{\ell=1}^it_\ell(G)=\sum_{\ell=1}^k t_{\ell}(G)\sum_{i=\ell}^k \delta_i(G) = \sum_{\ell=1}^k t_{\ell}(G)\big( n-\Delta_{\ell}\big),
\end{align*}
where we used that $\sum_{i=1}^k \delta_i(G) = n$ and so $\sum_{i=\ell}^k \delta_i(G) = n-\Delta_{\ell}$.
Summing over all possible outcomes of $\Delta_{\ell}$ for $\ell \in [k]$, by the law of total expectation, we have 
\begin{align*}
\E\left[\sum_{i=1}^n \widetilde{\tau}_i(G)\right] &=\sum_{\ell=1}^k\sum_{\delta=0}^{n-1}\E\Big[ t_{\ell}(G)\big( n-\Delta_{\ell}\big) \, \big| \, \Delta_{\ell}=\delta\Big]\cdot\Pr\Big[\Delta_{\ell}=\delta\Big] \\
&= \sum_{\ell=1}^k\sum_{\delta=0}^{n-1}\frac{n(n-\delta)}{n-\delta}\Pr\Big[\Delta_{\ell}=\delta\Big] \\
&= kn \sum_{\delta=0}^{n-1}\Pr\Big[\Delta_{\ell}=\delta\Big] = kn,
\end{align*}
where the second-to-last equality follows from~\eqref{eq:condprob}.
\end{proof}

The result of Theorem~\ref{thm:sumra} inspires the following definition.

\begin{definition} \label{def:recbal}
Let $\mC \subseteq \F_q^n$ be a code with generator matrix~$G$. 
We call $\mC$ a \textbf{recovery balanced} code if ${\E[\widetilde{\tau}_1(G)] = \dots = \E[\widetilde{\tau}_n(G)].}$
\end{definition}

Note that, by Claim~\ref{claim_good}, whether or not a code is recovery balanced does \textit{not} depend on the choice of the generator matrix $G$, and thus it is indeed a code property.

The following is an immediate consequence of Theorem~\ref{thm:sumra}.
\begin{corollary} \label{cor:alwaysk}
Let $\mC \subseteq \F_q^n$ be a recovery balanced code 
and let $G$ be the systematic generator matrix of $\mC$. We have $\E[{\tau}_i(G)]= k$ for all $i \in [k]$ and $\tmax(G) = k$.
\end{corollary}

\begin{remark} \label{rem:avvg}
Even though intuitively one would think that it is always better to consider a generator matrix in systematic form, instead of a generator matrix in non-systematic form, there does not seem to be an obvious way of proving this. However, through lengthy computations, one can find the average expected number of draws to recover the $i$-th information strand for some $i \in [k]$ (note that this average does not depend on the choice of the index $i$) over all systematic, and unrestricted, generator matrices, respectively. More formally, we computed
\begin{align*}
  \frac{\sum_{G \in \mG}\E[{\tau}_i(G)]}{|\mG|}  \qquad \textnormal{ and } \qquad \frac{\sum_{G \in \mG_{\operatorname{sys}}}\E[{\tau}_i(G)]}{|\mG_{\operatorname{sys}}|}
\end{align*}
where $\mG:=\{G \in \F_q^{k \times n}: \rk(G)=k\}$ and $\mG_{\operatorname{sys}} :=\{G \in \F_q^{k \times n}: \rk(G)=k, \; \textnormal{$G$ is systematic}\}$. The explicit expressions for these average values together with their proofs can be found in Appendix~\ref{sec:apen}, and we include a plot to illustrate the comparison between the two average values. It is evident from the plot in Figure~\ref{fig:averages} that \emph{on average} systematic generator matrices are favorable. However, computing the difference between the averages does not seem to be easy.
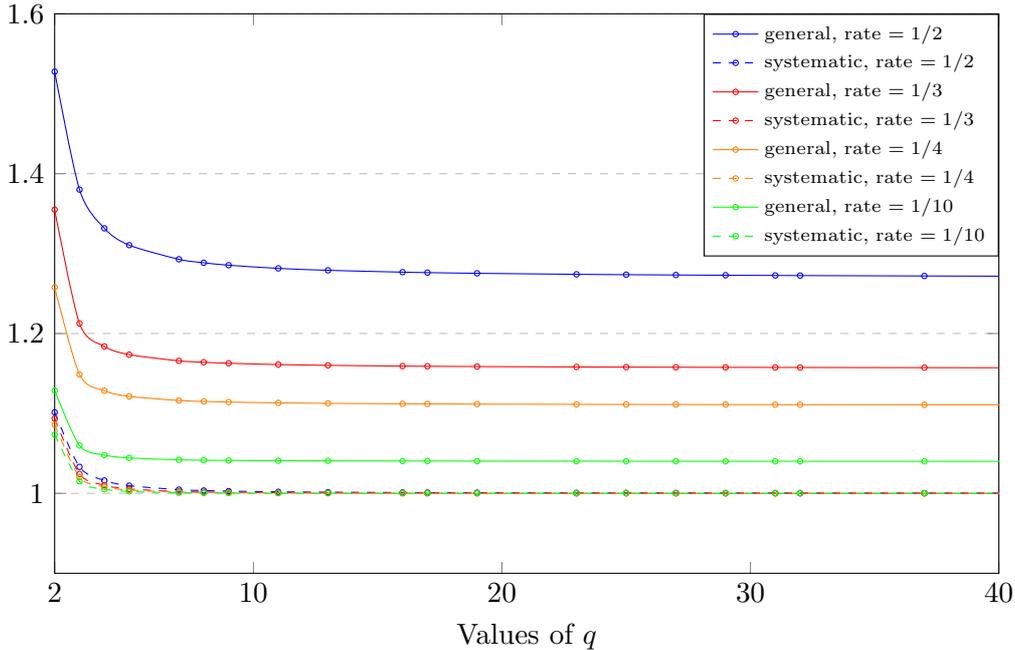
\begin{figure}[ht!]
\centering
\begin{tikzpicture}[scale=1]
\begin{axis}[legend style={at={(1,1)}, legend style={cells={align=left}}, anchor = north east, /tikz/column 2/.style={
                column sep=5pt}},
		legend cell align={left},
		width=14cm,height=9cm,
    xlabel={Values of $q$},
    xmin=2, xmax=40,
    ymin=0.9, ymax=1.6,
    xtick={2,10,20,30,40},
    ytick={1,1.2,1.4,1.6},
    ymajorgrids=true,
    grid style=dashed,
    every axis plot/.append style={},  yticklabel style={/pgf/number format/fixed}
]
\addplot+[color=blue,mark=o,mark size=1pt,smooth]
coordinates {
(2,1.52774151096388093803265104196)
(3,1.38007009178629426812279596476)
(4,1.33160423192185204864749820150)
(5,1.31058544225923601120457710330)
(7,1.29289619223208050622211458495)
(8,1.28854473422007195156331252619)
(9,1.28549448013483940974784937836)
(11,1.28153086102024841244155985503)
(13,1.27908449988711187514361087002)
(16,1.27679109030227367305194891732)
(17,1.27624111849744021376691524668)
(19,1.27534489212982711759445243414)
(23,1.27408608230798740474868586790)
(25,1.27362749515945839394193183678)
(27,1.27324511939029647350989679192)
(29,1.27292146569954804303271708642)
(31,1.27264400456806180839562108823)
(32,1.27251961988891971852811908883)
(37,1.27200747617078996853002570598)
(41,1.27169487718432187227336443851)
(43,1.27156205405279000010773497717)
(47,1.27133269049370239062680236262)
(49,1.27123299414255817084608855403)
(53,1.27105757473989013831294414800)
(59,1.27084152551718568921722974721)
(61,1.27077946354732401055314398375)
(64,1.27069401796107386503228412424)
(67,1.27061660056833598862872960531)
(71,1.27052402659912856157405873104)
(73,1.27048171590535302612666728120)
(79,1.27036817391335966779454314730)
(81,1.27033421682941707839160364039)
(83,1.27030196133908104716380605570)
(89,1.27021421538898080274792809427)
(97,1.27011468055033913962151523492)
    };
\addplot+[color=blue,mark=o,mark size=1pt,smooth,dashed,mark options={solid}]
coordinates {
(2,1.10137361798967633928571428571)
(3,1.03313093032767151499679350059)
(4,1.01608617558044248393603733608)
(5,1.00954482761007104000000000000)
(7,1.00461132268223244630355671950)
(8,1.00351973509469853951259210589)
(9,1.00279640175251856552536612796)
(11,1.00192170162414031345774042442)
(13,1.00142783768651414157834844742)
(16,1.00100485270869416091660722919)
(17,1.00091009118407839319675085260)
(19,1.00076184270657142499266107219)
(23,1.00056766865797436829021746684)
(25,1.00050135867722599055808204044)
(27,1.00044800187911573940363206443)
(29,1.00040426624980612810764820276)
(31,1.00036784860805642000701148816)
(32,1.00035185303768548712136269109)
(37,1.00028821410379702972829053351)
(41,1.00025118495278391964293692093)
(43,1.00023588143630880316135094620)
(47,1.00021007298751578809921794071)
(49,1.00019910302407689835729307093)
(53,1.00018017237679927801578258332)
(59,1.00015751932048089280126174003)
(61,1.00015114935915690081505250722)
(64,1.00014248090878988017931323388)
(67,1.00013472934772201977267199695)
(71,1.00012558922201244536393156993)
(73,1.00012145889695954008109539912)
(79,1.00011052254458283900999002338)
(81,1.00010729383620315102105641453)
(83,1.00010424495928935219359868231)
(89,1.00009604041527781606736842031)
(97,1.00008689318330501225196639914)
};
\addplot+[color=red,mark=o,mark size=1pt,smooth]
coordinates {
(2,1.35494026405756900360901584961)
(3,1.21261107399542487168333517869)
(4,1.18386014769622291825368673824)
(5,1.17359274883088885333030092340)
(7,1.16583081986421192365832197414)
(8,1.16402051821427865038477689990)
(9,1.16276832826853084847781881877)
(11,1.16115567229807590215447899755)
(13,1.16016468370034118978819770108)
(16,1.15923562560924899650506486246)
(17,1.15901251934518610325969951757)
(19,1.15864849319265955159976221989)
(23,1.15813593356320068162920110754)
(25,1.15794875421887240080317326283)
(27,1.15779246867537227408975219278)
(29,1.15766002192856227923943754892)
(31,1.15754635333139050282558806051)
(32,1.15749535731735404679976592767)
(37,1.15728512082652049714138981337)
(41,1.15715657996742233589415180441)
(43,1.15710191103769088271866849396)
(47,1.15700743202310992118362871808)
(49,1.15696633526781370004801159844)
(53,1.15689397913127982304854060702)
(59,1.15680478438522018337030944573)
(61,1.15677914593990311408347549202)
(64,1.15674383531529796240703274925)
(67,1.15671183012671378745891024472)
(71,1.15667354367822306757866753303)
(73,1.15665603935553627459659630141)
(79,1.15660904853152827421051499196)
(81,1.15659498999132476261310241519)
(83,1.15658163380313877219001791421)
(89,1.15654528991109992974996516719)
(97,1.15650404446536464340636138022)
    };
\addplot+[color=red,mark=o,mark size=1pt,smooth,dashed,mark options={solid}]
coordinates {
(2,1.09362495141209695826877247204)
(3,1.02418710213250541566549623007)
(4,1.01019346656679639066624471375)
(5,1.00545866931828670014310685425)
(7,1.00228059773547398215679581243)
(8,1.00164830909592530009662043525)
(9,1.00124993595224054605720141140)
(11,1.00079609066018828863006186600)
(13,1.00055710145125064393815273962)
(16,1.00036552704380773747028622059)
(17,1.00032465360359511941011930463)
(19,1.00026249565755028050921084537)
(23,1.00018489480291359950672841264)
(25,1.00015955373374828761595280220)
(27,1.00013965680798782239339636489)
(29,1.00012370598151730020799619662)
(31,1.00011069050428773011181271522)
(32,1.00010505509816792690234074705)
(37,1.00008316789923042597998376975)
(41,1.00007086215838500870822798361)
(43,1.00006587814993076488170392734)
(47,1.00005761800179444394243368185)
(49,1.00005416515563605551338986717)
(53,1.00004829329046615499719483797)
(59,1.00004142012498793997947905083)
(61,1.00003951932751757245683779522)
(64,1.00003695629391623732035164633)
(67,1.00003468816920176659711115183)
(71,1.00003204369771754639202083652)
(73,1.00003085965625419803995101294)
(79,1.00002775874430256308161744679)
(81,1.00002685304697004010069625140)
(83,1.00002600200201843560714378067)
(89,1.00002373263296015593467840104)
(97,1.00002123956438560007401791921)
};
\addplot+[color=orange,mark=o,mark size=1pt,smooth]
coordinates {
(2,1.25786053707200107352048319110)
(3,1.14875915126772760495832358027)
(4,1.12837387218214142103249479740)
(5,1.12133386610959394501924907432)
(7,1.11619414508855991657019143850)
(8,1.11502976580642042399790090113)
(9,1.11423349275931377706318771707)
(11,1.11321918131507104644024312787)
(13,1.11260180140573503595712485899)
(16,1.11202653538169389226278972830)
(17,1.11188881665203016971961377973)
(19,1.11166440306622158833144713860)
(23,1.11134892096397243694096984725)
(25,1.11123382144882040490009296284)
(27,1.11113775157702518015872477342)
(29,1.11105635363278330095543868273)
(31,1.11098650616416056615422447581)
(32,1.11095517223000874754074767748)
(37,1.11082600366031370434826745202)
(41,1.11074703079862746010145404159)
(43,1.11071344272243577334172417303)
(47,1.11065539376619088305830095881)
(49,1.11063014241022647721843066281)
(53,1.11058568214000142282053042914)
(59,1.11053087088136688088604695448)
(61,1.11051511473622962115866451682)
(64,1.11049341372407801927069085783)
(67,1.11047374328649419255959393667)
(71,1.11045021123494638643481190466)
(73,1.11043945209593413054596473407)
(79,1.11041056748584018998157039986)
(81,1.11040192548124043549828941082)
(83,1.11039371504110128102198097205)
(89,1.11037137246369336295106074411)
(97,1.11034601494306384708027088809)
    };
\addplot+[color=orange,mark=o,mark size=1pt,smooth,dashed,mark options={solid}]
coordinates {
(2,1.08591162314251783693497510375)
(3,1.02040604337228257283364180794)
(4,1.00805655239943871583940097735)
(5,1.00408549621578022607989046902)
(7,1.00156691432108722311039115412)
(8,1.00109440226724472856254991334)
(9,1.00080538576584795949493088185)
(11,1.00048781882172598777403983487)
(13,1.00032777295587572618432224506)
(16,1.00020484828370902013043933917)
(17,1.00017945586065607269837060636)
(19,1.00014155983304843051623979009)
(23,1.00009574187181598970732723069)
(25,1.00008122773567807605591291083)
(27,1.00007002050992674386117053071)
(29,1.00006117147280764139659253780)
(31,1.00005405064088227011041389276)
(32,1.00005099779996093655563509297)
(37,1.00003933601031091123419079789)
(41,1.00003293438072521014735762415)
(43,1.00003037817003431815394507090)
(47,1.00002619331578197083151774882)
(49,1.00002446461698079960125600646)
(53,1.00002155532697967952383785933)
(59,1.00001820341717035385346193796)
(61,1.00001728756367905854114857790)
(64,1.00001606085379130205710053720)
(67,1.00001498355796866004729576082)
(71,1.00001373787391707976069482032)
(73,1.00001318392387833506859230247)
(79,1.00001174496186236540754683133)
(81,1.00001132804815595374641909353)
(83,1.00001093774384639815142721056)
(89,1.00000990412008029299993044218)
(97,1.00000878130069917961627237808)
};
\addplot+[color=green,mark=o,mark size=1pt,smooth, solid]
coordinates {
(2,1.12859684771229293609439871817)
(3,1.05984646488019740220210424029)
(4,1.04792426616760386115335053805)
(5,1.04424703104997234176574817094)
(7,1.04191653781495314800155636983)
(8,1.04145607722432934520694034933)
(9,1.04115956977187431946646454957)
(11,1.04080484455573078891442903362)
(13,1.04060147267048433432962303267)
(16,1.04041982405218523424167378928)
(17,1.04037733841147796767245513770)
(19,1.04030883504935937771660754502)
(23,1.04021386549002070655066195824)
(25,1.04017954481432250569013001482)
(27,1.04015101111245838258268030506)
(29,1.04012690642005942052602591084)
(31,1.04010626874068836468104637827)
(32,1.04009702327870788367472965204)
(37,1.04005898322024743985370622409)
(41,1.04003577491162842529761946343)
(43,1.04002591350266872018302634913)
(47,1.04000888168635273599285775874)
(49,1.04000147678416810760878105450)
(53,1.03998844395188023938342139149)
(59,1.03997238443112715896500189143)
(61,1.03996776923356795428720739703)
(64,1.03996141351644839772973891154)
(67,1.03995565324238233104663354075)
(71,1.03994876294254750208502645815)
(73,1.03994561287332504221336823669)
(79,1.03993715670242909949994901772)
(81,1.03993462686567922626263836753)
(83,1.03993222342691739271876395007)
(89,1.03992568336457667758084530496)
(97,1.03991826115932137647966735498)
    };
\addplot+[color=green,mark=o,mark size=1pt,smooth,dashed,mark options={solid}]
coordinates {
(2,1.07342343932673194334841638894)
(3,1.01516021711211302414573205729)
(4,1.00525393748116944593101811348)
(5,1.00236873145205420290580901139)
(7,1.00074357236400875583402147869)
(8,1.00047671382693749511017685632)
(9,1.00032457749445186839295559178)
(11,1.00017160382000206132477293481)
(13,1.00010275421715739803057683997)
(16,1.00005561165469991632656140052)
(17,1.00004672037499553529529416465)
(19,1.00003414387952702131265282912)
(23,1.00002029305561699406637668117)
(25,1.00001628942684767379553944933)
(27,1.00001335135107346420864651436)
(29,1.00001113748353081397576001801)
(31,1.00000943128721219006626967135)
(32,1.00000872209229569206783291569)
(37,1.00000614984380781375282187222)
(41,1.00000484151584853125797708011)
(43,1.00000434242205987732956274444)
(47,1.00000355697643903384172042941)
(49,1.00000324480562206939182040168)
(53,1.00000273699182201674354411295)
(59,1.00000218151667005423354824129)
(61,1.00000203570826233211699349589)
(64,1.00000184469792663611208044775)
(67,1.00000168117882679821455068961)
(71,1.00000149728053023806801782183)
(73,1.00000141736478978851550229983)
(79,1.00000121541801244480629358253)
(81,1.00000115849467635124893376034)
(83,1.00000110587720861003580326124)
(89,1.00000096978858624832634837725)
(97,1.00000082759777925743863215285)
};
\legend{\tiny{general, rate $=1/2$}, \tiny{systematic, rate $=1/2$},\tiny{general, rate $=1/3$}, \tiny{systematic, rate $=1/3$},\tiny{general, rate $=1/4$}, \tiny{systematic, rate $=1/4$},\tiny{general, rate $=1/10$}, \tiny{systematic, rate $=1/10$}
}
\end{axis}
\end{tikzpicture}
\caption{\label{fig:averages} Normalized random access expectation from the formulas in Appendix~\ref{sec:apen} for $k=4$ and various rates (where the rate is the ratio between the dimension $k=4$ and the length $n$).}
\end{figure}
Note that for the average for systematic generator matrices, the ratio between the random access expectation and the dimension of the considered codes goes to 1 as $q$ grows. This can be explained by the fact that for large $q$ MDS-matrices are dense within the set of rank-$k$ matrices, and from Corollary~\ref{cor:mds} we know that MDS-matrices have expectation~$k$. Furthermore, in~\cite{flajolet1992birthday} it was shown that the expected number of draws to retrieve any size-$k$ subset of the $n$ columns is $n(H_n-H_{n-k})$. It is interesting to observe that from experiments it seems that the ratio of the non-systematic case approaches ${n(H_n-H_{n-k})}/{k}$. This indicates that for the non-systematic case, on average, one cannot usually do better than retrieving all the $k$ information strands in order to retrieve only the $i$-th information strand.
\end{remark}

\section{Recovery Balanced Codes}\label{sec:rec_bal}
This section studies more in detail codes that are recovery balanced; see Definition~\ref{def:recbal}. We give three different sufficient conditions for a code to be recovery balanced, one of which works mainly for binary codes, whereas the other two do not have any restrictions on $q$.
Using these conditions we show that various families of codes (such as MDS codes, Hamming codes, simplex codes, binary Reed-Muller codes, binary Golay codes) are recovery balanced. This allows us to compute the value of
$\tmax(G)$ for any of their systematic generator matrices by Corollary~\ref{cor:alwaysk}.

\subsection{Codes with transitive permutation automorphism group}
The condition for a code to be recovery balanced presented in this subsection is connected to its automorphism group; see e.g.~\cite{macwilliams1977theory}. Before presenting this condition, we will recall the necessary definitions.

For $\sigma \in S_n$ (the symmetric group of order $n$) and $x \in \F_q^n$, let $f_\sigma(x)$ be defined by $f_\sigma(x):=(x_{\sigma(1)}, \ldots, x_{\sigma(n)})$ for all $x \in \F_q^n$. For a code $\mC \subseteq \F_q^n$ we let $f_{\sigma}(\mC) = \{f_{\sigma}(x) : x \in \mC\}$.

\begin{definition}
    For a linear code $\mC \subseteq \F_q^n$ the \textbf{permutation automorphism group} of $\mC$ is $\mbox{PAut}(\mC)=\{\sigma \in S_n : f_\sigma(\mC)=\mC\}$. $\mbox{PAut}(\mC)$ is called \textbf{transitive} if for each $i,j \in [n]$ there exists $\sigma \in \mbox{PAut}(\mC)$ with $\sigma(i)=j$.
\end{definition}


\begin{lemma} \label{lem:auto}
    Let $\mC \subseteq \F_q^n$ be a linear code. Suppose there exists $\sigma \in \mbox{PAut}(\mC)$ and $i,j \in [n]$ with
    $\sigma(i)=j$. Then $\E[\widetilde{\tau}_i(G)]=\E[\widetilde{\tau}_j(G)]$ for any generator matrix $G$ of $\mC$.
\end{lemma}
\begin{proof}
Let $G$ be a generator matrix of $\mC$, where we denote by $g_{\ell}$ the $\ell$-th column of $G$ for $\ell \in [n]$. Since $\sigma \in \mbox{PAut}(\mC)$ we have that 
\begin{align*}
    \sigma(G) := (g_{\sigma(1)},\dots, g_{\sigma(n)}) \in \F_q^{k \times n}
\end{align*}
is a generator matrix for $\mC$ as well. 
We have that $\widetilde{\alpha}_i(s)$ does not depend on the choice of the generator matrix $G$, but only on the code $\mC$, for all $s \in [n]$; see the proof of Claim~\ref{claim_good}. 
Therefore by computing $\widetilde{\alpha}_i(s)$ with the two generator matrices we have
\begin{align*}
    \widetilde{\alpha}_i(s) = |\{S \subseteq [n] : |S| = s, \, g_i \in \langle g_{\ell} : {\ell} \in S\rangle\}| =
     |\{S \subseteq [n] : |S| = s, \, g_{\sigma(i)} \in \langle g_{\sigma({\ell})} : {\ell} \in S\rangle\}|
\end{align*}
for all $s$.
Since $\sigma(i)=j$, we have
\begin{align*}
   |\{S \subseteq [n] : |S| = s, \, g_{\sigma(i)} \in \langle g_{\sigma({\ell})} : {\ell} \in S\rangle\}| &= |\{S \subseteq [n] : |S| = s, \, g_j \in \langle g_{\sigma({\ell})} : {\ell} \in S\rangle\}|  \\
   &= |\{S \subseteq [n] : |S| = s, \, g_j \in \langle g_{\ell} : {\ell} \in S\rangle\}|   
\end{align*}
for all $s$, where the latter equality follows from the fact that $\sigma$ is bijective and thus $\{g_1,\dots,g_n\} = \{g_{\sigma(1)},\dots,g_{\sigma(n)}\}$ as multisets. In more detail, we use that the number of subsets of size $s$ of $\{g_1,\dots,g_n\}$ that have $g_j$ in their span is the same as the number of subsets of size $s$ of the same set $\{g_{\sigma(1)},\dots,g_{\sigma(n)}\}$ that have $g_j$ in their span. From a straightforward modification of Lemma~\ref{lem:fi}, stated in~\eqref{eq:aaa}, it then immediately follows that $\E[\widetilde{\tau}_i(\mC)]=\E[\widetilde{\tau}_j(\mC)]$.
\end{proof}

The next corollary follows directly from Lemma~\ref{lem:auto}.
\begin{corollary} \label{cor:paut}
    If $\mbox{PAut}(\mC)$ is transitive, then $\mC$ is recovery balanced.
\end{corollary}
\begin{proof}
    Since for any $i,j \in [n]$ there exists $\sigma \in \mbox{PAut}(\mC)$ with the property that $\sigma(i)=j$, the statement follows immediately from Lemma~\ref{lem:auto}.
\end{proof}

Note that the reverse statement of Corollary~\ref{cor:paut} is not true in general. For example, the permutation automorphism group of non-binary MDS codes is not necessarily transitive, although we know that MDS codes are recovery balanced; see Proposition~\ref{prop:fimds}. We illustrate this with the following example. 


\begin{example}
Consider the MDS code $\mathcal{C}$ over $\mathbb{F}_3$ with the generator matrix 
\begin{align*}
G = \begin{pmatrix}
1 & 0 & 1 & 1 \\
0 & 1 & 1 & 2
\end{pmatrix} = (g_1, g_2, g_3, g_4) \in \mathbb{F}_3^{2 \times 4},
\end{align*}
where $g_{\ell}$ denotes the $\ell$-th column of $G$.
There is no $\sigma \in \mathrm{PAut}(\mathcal{C})$ with $\sigma(1) = 2$. Indeed, the only generator matrices of $\mC$ with the second column of $G$ as their first column are those of the form
$$\begin{pmatrix}
0 & \beta & \beta & 2\beta \\
1 & 0 & 1 & 1
\end{pmatrix}$$
with $\beta \in \F_3\setminus \{0\}$. None of these matrices have the same set of columns as $G$.
\end{example}

Corollary~\ref{cor:paut} solves the problem of computing $\tmax(G)$ in the case where~$G$ is the systematic generator matrix of a code with transitive permutation automorphism group. Some examples of such codes include binary Reed-Muller codes (see~\cite[Chapter 13, Section 9]{macwilliams1977theory}), binary Golay codes (see~\cite[Chapter 20, Section 1]{macwilliams1977theory}), and also binary Hamming and simplex codes. From this we obtain the following.

\begin{corollary}
    Binary simplex codes, binary Hamming codes, binary Reed-Muller codes and binary Golay codes are all recovery balanced. 
\end{corollary}

\begin{remark}
In classical coding theory, codes with certain symmetries simplify the analysis and implementation of encoding and decoding algorithms, often leading to more efficient and predictable performance. From this section it is evident that the property of being recovery balanced is a notion of symmetry in a code, generalizing for example the property of having transitive permutation automorphism group. We want to highlight that the random access problem is an instance for when codes with too much symmetry are \emph{bad}.
\end{remark}

\subsection{Two sufficient conditions for codes to be recovery balanced}
In this subsection we give two more sufficient conditions for a code to be recovery balanced, which work over any finite field $\F_q$, and which we apply in the sequel. 
The following lemma follows immediately from the definition of recovery
balance codes and 
the natural analogues of Lemmas~\ref{lem:fi} and~\ref{lem:bsj}, and it provides two code properties that imply recovery balance.

\begin{lemma} \label{lem:recba}
Let $\mC \subseteq \F_q^n$ be a code with generator matrix~$G$.
\begin{itemize}
    \item[(i)] Let $\widetilde{\alpha}_i(s)={|\{S \subseteq [n] : |S| = s, \, g_i \in \langle g_j : j \in S\rangle \}|}$. If the sequence $(\widetilde{\alpha}_i(s) : s \in [n])$ is the same for all $i \in [n]$, then $\mC$ is recovery balanced. 
    \item[(ii)] Let $\widetilde{\mR}(i)=\{R_1,\dots,R_L\}$ be the recovery sets for the $i$-th encoded strand
    and let $\widetilde{\beta}_i(s,j) = \left| \{ S \subseteq [L] : |S|=s, \,|\bigcup_{h \in S} R_h | = j  \}  \right|$.
    If the sequence $(\widetilde{\beta}_i(s,j) : s \in [L], j \in [n])$ is the same for all $i \in [n]$, then $\mC$ is recovery balanced. 
\end{itemize}
\end{lemma}

We show how one can apply Lemma~\ref{lem:recba} to three classes of codes. 
We start with a result whose proof is analogous 
to that of Corollary~\ref{cor:mds}. 

\begin{proposition} \label{prop:fimds}
Let $\mC \subseteq \F_q^n$ be an MDS code of dimension~$k$. For all $i,s \in [n]$ we have
\begin{align*}
    \widetilde{\alpha}_i(s) = \begin{cases}
        \binom{n-1}{s-1} \quad &\textnormal{if $s \in [k-1]$,} \\
        \binom{n}{s} \quad &\textnormal{if $s \ge k$,}
    \end{cases}
\end{align*}
In particular, $\mC$ is recovery balanced and $\tmax(G)=k$ if $G$ is systematic. 
\end{proposition}


We now turn to $q$-ary Hamming and simplex codes. We show that they are recovery balanced, which allows us to compute the value
of $\tmax(G)$ for any of their systematic generator matrices.

\begin{proposition}\label{prop:betahamm}
Let $\mC \subseteq \F_q^n$ be the $q$-ary Hamming code of redundancy $k$. For $i \in [n]$, $s \in [q^{k-1}+1]$ and $j \in [(q^k-1)/(q-1)]$ we have
\begin{align*}
    \widetilde{\beta}_i(s,j)= 
    \begin{cases}
        \gamma(s,v) \quad &\textnormal{if $j=(q^k-q^v)/(q-1)-1$,} \\
        \gamma(s-1,v) \quad &\textnormal{if $j=(q^k-q^v)/(q-1)$,}
    \end{cases}
\end{align*} 
where $\gamma(s,v)$ is equal to
\begin{align*}
    \qbin{k-1}{v}{q}\sum_{u=v}^{k-1} q^{u}\qbin{k-v-1}{u-v}{q}
    \binom{q^{k-u-1}}{s}(-1)^{u-v}q^{\binom{u-v}{2}}.
\end{align*}
In particular, $\mC$ is recovery balanced and $\tmax(G)=k$ if $G$ is a systematic generator matrix of~$\mC$.
\end{proposition}
\begin{proof}
Let $H$ be a parity-check matrix of $\mC$. A set $S \subseteq [n]$ is a recovery set (of size $\ge 2$) for the $i$-th encoded strand if and only if there exists $x \in \mC^\perp$ with $i \in \supp(x) \subseteq S \cup \{i\}$; see Claim~\ref{claim:dual}. Moreover, since $H$ has as columns all the non-zero vectors (up to multiples) of $\F_q^k$, $S \subseteq [n]$ is the support of some $x \in \mC^\perp$ if and only if the columns of $H$ indexed by~$S^c$ form a hyperplane of $\F_q^k$.

    In the remainder of the proof, we denote by $\rho(V)$ the set of 1-dimensional subspaces of a space $V \le \F_q^k$ and we let $h_i$ be the $i$-th column of $H$. Note that if $\dim(V)=v$ then $|\rho(V)|=(q^v-1)/(q-1)$. In order to give an explicit formula for~$\widetilde{\beta}_i(s,j)$, we count the number of sets of~$s$ (distinct) hyperplanes $\{\mH_1,\dots,\mH_s\}$ with the properties that $h_i \notin \mH_\ell$ for all $\ell \in [s]$ and $|\bigcup_{\ell=1}^s \rho(\mH_\ell)^c|=j$. Using simple set theory, this is equivalent to asking that $|\bigcap_{\ell=1}^s \rho(\mH_\ell)|=\rho(\F_q^k)-j=(q^k-1)/(q-1)-j$. Note that we have $\bigcap_{\ell=1}^s \rho(\mH_\ell)=\rho\left(\bigcap_{\ell=1}^s \mH_\ell\right) = (q^v-1)/(q-1)$ for some integer $v \in [k]$. 

In the rest of the proof we use
the Möbius Inversion formula for the lattice of subspaces (see e.g.~\cite[Propositions 3.7.1 and Example~3.10.2]{stanley2011enumerative}).
Let $\mP$ be the collection of hyperplanes of~$\F_q^k$.
For a subspace $V \le \F_q^k$ define
\begin{align*}
    f(V) := \Bigl|\Bigl\{\mH =\{\mH_1, \ldots, \mH_s\} \subseteq \mP : |\mH|=s,  h_i \notin H_\ell \textnormal{ for all $\ell \in [s]$}, \, \bigcap_{\ell=1}^s \mH_\ell = V \Bigr\}\Bigr|
\end{align*}
and 
\begin{align*}
    g(V) := 
    \Bigl|\Bigl\{&\mH= \{\mH_1,\dots,\mH_s\} \subseteq \mP : |\mH|=s, h_i \notin \mH_\ell
    \textnormal{ for all $\ell \in [s]$}, \, \bigcap_{\ell=1}^s \mH_\ell \ge V\Bigr\}\Bigr|
    \\
    =\Bigl|\Bigl\{&\{\mH_1,\dots,\mH_s\} \subseteq \mP : |\mH|=s, h_i \notin \mH_\ell \mbox{ and } V \le \mH_\ell \textnormal{ for all $\ell \in [s]$}\Bigr\}\Bigr|.
\end{align*}

We continue by giving an explicit expression for $g(V)$ in terms of $v=\dim(V)$. If $h_i \notin V$ we have that the number of hyperplanes $\mH_{\ell}$ with~$h_i \notin \mH_{\ell}$ and $V \le \mH_{\ell}$ is 
\begin{align*}
    \qbin{k-v}{k-1-v}{q}-\qbin{k-v-1}{k-1-v-1}{q}.
\end{align*}
Using the well-known identity for $q$-ary binomial coefficients
\begin{align} \label{idee}
   \qbin{a}{b}{q}-\qbin{a-1}{b-1}{q} =  q^b\qbin{a-1}{b}{q} \; \textnormal{for all $a\ge b \ge 1$,}
\end{align} 
we then have 
\begin{align*}
    g(V)&=\sum_{U \ge V} f(U) = \begin{cases}
        0 \quad &\textnormal{if $h_i \in V$,} \\
        \binom{q^{k-v-1}}{s}\quad &\textnormal{if $h_i \notin V$.}
    \end{cases}
\end{align*}
Therefore, by the Möbius Inversion formula, we have
\begin{align*}
    f(V) =& \sum_{U \ge V} g(U)(-1)^{\dim(U)-v}q^{\binom{\dim(U)-v}{2}} \\
    =& \sum_{\substack{U \ge V \\ h_i \notin U}} g(U)(-1)^{\dim(U)-v}q^{\binom{\dim(U)-v}{2}} \\
    =&\sum_{u=v}^{k-1} \left(\qbin{k-v}{u-v}{q} -\qbin{k-v-1}{u-v-1}{q} \right)\binomi{q^{k-u-1}}{s}(-1)^{u-v}q^{\binom{u-v}{2}}. 
\end{align*}
This, together with
\eqref{idee} again,
after summing over all $V\le \F_q^k$ with $\dim(V)=v$ and with $h_i \notin V$ (there are $\qbin{k}{v}{q}-\qbin{k-1}{v-1}{q}=q^v\qbin{k-1}{v}{q}$ such spaces), shows that
$\gamma(s,v)$ counts the number of $s$-sets of codewords in $\mC^\perp$ with the property that the union of their support has cardinality $(q^k-q^v)/(q-1)$ and which all contain~$i$ in their support. 

To conclude the proof, note that the only recovery sets we did not consider so far are
those of size 1.
For each $i$, distinguishing between the case where this recovery set is one of the $s$ sets considered in the computation of $\widetilde{\beta}_i(s,j)$ and the case where this is not the case, gives the formula in the proposition.
\end{proof}


\begin{proposition}\label{prop:simpfi}
Let $\mC \subseteq \F_q^n$ be the $q$-ary simplex code of dimension $k$. For all $i,s \in [n]$ we have
\begin{align*}
    \widetilde{\alpha}_i(s) = \sum_{d=1}^{s}\qbin{k-1}{d-1}{q} \sum_{r=1}^d \qbin{d}{r}{q} \binom{\frac{q^r-1}{q-1}}{s}(-1)^{d-r}q^{\binom{d-r}{2}}.
\end{align*}
In particular, $\mC$ is recovery balanced and $\tmax(G)=k$ if $G$ is a systematic generator matrix of~$\mC$.
\end{proposition}
\begin{proof}
A generator matrix of $\mC$ has as columns all the non-zero vectors (up to multiples) in $\F_q^k$. For $S \subseteq [n]$ we denote by $G_S$ the set of columns of $G$ indexed by~$S$, i.e., $G_S = \{ g_j : j \in S\} $.
We have
\begin{align*}
     |\{S \subseteq [n] \colon |S|=s, \, g_i \in \langle G_S \rangle \}| = 
     \sum_{d=1}^s |\{S \subseteq [n] \colon |S|=s, \, g_i \in G_S, \, \dim \langle G_S \rangle  = d \}|.
\end{align*}
In order to compute $|\{S \subseteq [n] \colon |S|=s, \, g_i \in \langle G_S \rangle , \, \dim \langle G_S \rangle  = d \}|$ for $d \in [s]$ we use the Möbius Inversion formula in the lattice of subspaces in $\F_q^k$. For a subspace $V \le \F_q^k$ we let
\begin{align*}
    f(V) &= |\{S \subseteq [n] \colon |S|=s, \langle G_S \rangle  =V\}|, \\
    g(U) &= \sum_{U \le V} f(U) = |\{S \subseteq [n] \colon |S|=s, \langle G_S \rangle  \le V\}| =\binom{\frac{q^{\dim(V)}-1}{q-1}}{s}.
\end{align*}
Using the Möbius Inversion formula we obtain
\begin{align*}
 f(V) &= \sum_{U \le V} g(U)(-1)^{\dim(V)-\dim(U)}q^{\binom{\dim(V)-\dim(U)}{2}} \\
 &= \sum_{r=1}^{\dim(V)}\sum_{\substack{U \le V, \\ \dim(U)=r}} \binom{\frac{q^r-1}{q-1}}{s}(-1)^{\dim(V)-r}q^{\binom{\dim(V)-r}{2}} \\
  &=\sum_{r=1}^{\dim(V)}\qbin{\dim(V)}{r}{q} \binom{\frac{q^{r}-1}{q-1}}{s}(-1)^{\dim(V)-r}q^{\binom{\dim(V)-r}{2}}.  
\end{align*}
In particular, we have
\begin{align*}
    |\{S \subseteq [n] \colon |S|=s, \, g_i \in \langle G_S \rangle , \, \dim \langle G_S \rangle  = d \}|
    &= \sum_{\substack{V \le \F_q^k, \\ \dim(V)=d, \\ g_i \in V}} f(V) \\
    &= \qbin{k-1}{d-1}{q} \sum_{r=1}^d \qbin{d}{r}{q} \binom{\frac{q^{r}-1}{q-1}}{s}(-1)^{d-r}q^{\binom{d-r}{2}},
\end{align*}
which concludes the proof.    
\end{proof}

\subsection{New codes from old}

In this subsection, we explore which code operations preserve the property of being recovery balanced.

Note that it is well-known, and easy to see, that for any code $\mC$ we have $\mbox{PAut}(\mC) = \mbox{PAut}(\mC^\perp)$. Hence we have the following result.

\begin{proposition} \label{prop:dualll}
    Let $\mC$ be a code with transitive permutation automorphism group. Then both~$\mC$ and $\mC^\perp$ are recovery balanced.
\end{proposition}

While Proposition~\ref{prop:dualll} only covers codes with transitive permutation automorphism group, we strongly believe that the property of being recovery balanced is closed under duality. This claim is supported for example by the fact that MDS codes, Hamming codes, and simplex codes are all recovery balanced over any finite field $\F_q$. Motivated by this, we propose the following conjecture.
\begin{conjecture} \label{conjecture}
A code $\mC$ is recovery balanced if and only if its  dual code $\mC^\perp$ is recovery balanced.
\end{conjecture}
An operation on codes that preserves their recovery balanced property is the Cartesian product. Although this is not difficult to demonstrate, we include the explanation here for completeness. Recall that for an $[n,k]_q$-code $\mC$ and an $[n',k']_q$-code $\mC'$ the Cartesian product of $\mC$ and $\mC'$ is  $\mC \times \mC':=\{c \circ c' : c \in \mC, c' \in \mC'\}$, where for $c=(c_1,\dots,c_n)\in \F_q^n$ and $c'=(c'_1,\dots,c'_n) \in \F_q^{n'}$, $c \circ c' := (c_1,\dots,c_n,c'_1,\dots,c'_n) \in \F_q^{n + n'}$. If $G$ is a generator matrix of $\mC$, and $G'$ is a generator matrix of $\mC'$, then $\mC \times \mC'$ has a generator matrix of the form
\begin{align*}
   G_{\mC\times\mC'} :=  \begin{pmatrix}
        G & \overline{\textbf{0}} \\
        \overline{\textbf{0}} & G'
    \end{pmatrix} \in \F_q^{(k+k') \times (n+n')},
\end{align*}
where $\overline{\textbf{0}}$ is the 0-submatrix with appropriate size. One can obtain the generator matrix of the Cartesian product of more than two codes analogously, resulting in a block matrix with the respective generator matrices on the diagonal. 

The proof of the following result uses very similar reasoning to the proof of~\cite[Proposition~11]{bar2023cover}. However, for the completeness of the results and the paper, we include it here. 

\begin{proposition} \label{prop:cart}
    For integers $1\le j\le t$, let $\mC_j$ be a recovery balanced $[n_j,k_j]_q$-code. For any $R\in(0,1]$, if for any $1\le j\le t$ we have that $k_j/n_j=R$, then the code $\mC_1 \times \mC_2 \times \cdots \times \mC_t$ is recovery balanced.
\end{proposition}
\begin{proof}
Denote by $G_j$ the generator matrix of $\mC_j$ and by $G_\times\triangleq G_{\mC_1 \times \mC_2 \times \cdots \times \mC_t}$ the generator matrix of $\mC_1 \times \mC_2 \times \cdots \times \mC_t$. For any $r \ge 1$ draws, let us denote by~$\varepsilon_j^r$  the random variable that governs the number of (not necessarily distinct) columns drawn from the  $n_j$ columns of $G_\times$ that corresponds to $G_j$ in the first $r$ draws. Additionally, let $n\triangleq n_1+n_2+\ldots+n_t$. For any $1\le j\le n$ and any $n_{j-1}< i\le n_j$, where $n_{0}\triangleq 0$, we have that
\begin{align*}
\E [\widetilde{\tau}_i (G_{\times})] & = \sum_{r=1}^\infty \Pr[ \widetilde{\tau}_i(G_{\mC\times\mC'}) \ge r] \\ 
& 
= \sum_{r=1}^\infty \sum_{z=0}^{\infty} \Pr [\varepsilon^{r-1}_j = z]\cdot \Pr\left[\widetilde{\tau}_i(G_{\times}) \ge r \mid \varepsilon^{r-1}_j = z\right]
\\
& \stackrel{(a)}{=} \sum_{r=1}^\infty \sum_{z=0}^{r-1} \Pr[\varepsilon^{r-1}_j = z]\cdot \Pr\left[\widetilde{\tau}_i(G_{\times}) \ge r \mid \varepsilon^{r-1}_j = z\right] \\ 
& = \sum_{r=1}^\infty \sum_{z=0}^{r-1} \binom{r-1}{z} \left(\frac{n_j}{n}\right)^z \left(\frac{n-n_j}{n}\right)^{r-z-1} \Pr\left[\widetilde{\tau}_i(G_{\times}) \ge r \mid \varepsilon_j^{r-1} = z\right]\\
& =   \sum_{r=1}^\infty \sum_{z=0}^{r-1} \binom{r-1}{z} \left(\frac{n_j}{n}\right)^z \left(\frac{n-n_j}{n}\right)^{r-z-1}  \Pr \left[ \widetilde{\tau}_i(G_j) \ge z+1\right] \\
& = \sum_{z=0}^\infty \Pr \left[ \widetilde{\tau}_i (G_j ) \ge z+1 \right] \sum_{r=z+1}^{\infty} \binom{r-1}{z} \left(\frac{n_j}{n}\right)^z \left(\frac{n-n_j}{n}\right)^{r-z-1}  \\ 
& = \sum_{z=0}^\infty \Pr \left[ \widetilde{\tau}_i (G_j ) \ge z+1 \right] \sum_{r=z}^{\infty} \binom{r}{z} \left(\frac{n_j}{n}\right)^z \left(\frac{n-n_j}{n}\right)^{r-z} \\ 
& \stackrel{(b)}{=} \sum_{z=0}^\infty \Pr \left[ \widetilde{\tau}_i (G_j ) \ge z+1 \right]  \cdot\frac{n}{n_j}  \\ 
& = \sum_{z=1}^{\infty} \Pr \left[ \widetilde{\tau}_i (G_j) \ge z \right] \cdot \frac{n}{n_j} = \frac{n}{n_j} \cdot \E \left[ \widetilde{\tau}_i (\mC_j) \right] \\
& \stackrel{(c)}{=} \frac{nk_j}{n_j}=nR,
\end{align*}
where equality (a) follows from the fact that the probability to collect $z>r-1$ columns from~$G_{\times }$, using only $r-1$ draws is zero for any integer $z$, i.e., $\Pr[\varepsilon_j^{r-1}=z]=0$.  
To see that equality (b) holds, recall that $\sum_{r=0}^\infty x^r=\frac{1}{1-x}$, and by taking the derivative of the latter $z$ times we get $$\sum_{r=z}^\infty r\cdot(r-1)\cdots(r-z+1)x^{r-z}=\frac{z!}{(1-x)^{z+1}},$$ which is equivalent to 
\begin{align}\label{eq:binomsums}
    \sum_{r=z}^\infty \binom{r}{z}x^{r-z}(1-x)^{z} = \frac{1}{1-x}.
\end{align}
Hence, by substituting $x=\frac{n-n_j}{n}$, equality (b) follows. Lastly, as $\mC_j$ is recovery balanced with dimension $k_j$, equality (c) holds. 
Thus for any $1\le i\le n$ we obtain 
\begin{align*}
    \E [\widetilde{\tau}_i (G_{\times})] = nR,
\end{align*}
 so the code $\mC_1 \times \mC_2 \times \cdots \times \mC_t$ is recovery balanced.
\end{proof}

Note that if we consider the  $\mC_j$'s to have different rates, the statement of Proposition~\ref{prop:cart} is not true in general. We include an example to illustrate this.
\begin{example}
    Let $\mC_1$ and $\mC_2$ be the codes over $\F_2$ with generator matrices $G_1$ and $G_2$, respectively, defined as follows:
    \begin{align*}
        G_1 := \begin{pmatrix}
            1 & 0 & 1 \\
            0 & 1 & 1
        \end{pmatrix}, \quad 
        G_2 := \begin{pmatrix}
            1 & 0  \\
            0 & 1
        \end{pmatrix}.
    \end{align*}
    Then $\mC_1$ has rate $2/3$, and $C_2$ has rate $1$. Both $\mC_1$ and $\mC_2$ are recovery balanced  as $\mC_1$ is an MDS code, and $\mC_2$ is the identity code. However, one can check that the code $\mC_1 \times \mC_2$ is not recovery balanced. In fact, in the language of Proposition~\ref{prop:cart}, we have
    \begin{align*}
        \E[\Tilde{\tau}_1(G_{\times})] = \E[\Tilde{\tau}_2(G_{\times})] = \E[\Tilde{\tau}_3(G_{\times})] = \frac{10}{3}, \quad \E[\Tilde{\tau}_4(G_{\times})] = \E[\Tilde{\tau}_5(G_{\times})] = \E[\Tilde{\tau}_6(G_{\times})] = 5.
    \end{align*}
\end{example}




It is natural to ask whether the sum and intersection of recovery balanced codes result in a recovery balanced code. This is not true in general, and we provide an example (Example~\ref{ex:summi}) where the sum does not preserve the property of being recovery balanced. Under the assumption that Conjecture~\ref{conjecture} holds, this would imply that, in general, the intersection of codes does not preserve the property of being recovery balanced either. Specifically, suppose we have codes $\mC$ and $\mC'$ that are recovery balanced, but where $\mC + \mC'$ is not. Then also $(\mC + \mC')^\perp = \mC^\perp \cap \mC'^\perp$ is not recovery balanced. However, assuming that Conjecture~\ref{conjecture} holds, both $\mC^\perp$ and $\mC'^\perp$ would be recovery balanced.

\begin{example} \label{ex:summi}
Let $q=13$, $n=9$ and $k=4$. Define the generator matrix 
\begin{align*}
    G= \begin{pmatrix}
        1 & 1 & 1 & 1 & 1 & 1 & 1 & 1 & 1 \\
        1 & 3 & 9 & 2 & 6 & 5 & 4 & 12 & 10 \\ 
        1 & 1 & 1 & 8 & 8 & 8 & 12 & 12 & 12 \\
        1 & 3 & 9 & 3 & 9 & 1 & 9 & 1 & 3
    \end{pmatrix} \in \F_{13}^{4 \times 9}.
\end{align*}
From a computer algebra program, we checked that the code $\mC$ generated by this matrix is recovery balanced. Note that this code was taken from~\cite[Example 1]{tamo2014family} and it is an optimal locally recoverable code.
Consider also the $4$-dimensional Reed-Solomon code $\mC'$ with evaluation vector $\alpha=(0,1,2,3,4,5,6,7,8)$ over $\F_{13}$. Since MDS codes are recovery balanced, this code is recovery balanced.
However, the sum of the two codes $\mC+\mC'$ has as generator matrix 
\begin{align*}
G_{\mC+\mC'}:= \begin{pmatrix}
 1 & 0 & 0 & 0 & 0 & 0 & 0 & 12 & 6 \\
 0 & 1 & 0 & 0 & 0 & 0 & 0 & 7 & 9 \\
 0 & 0 & 1 & 0 & 0 & 0 & 0 & 1 & 1 \\
 0 & 0 & 0 & 1 & 0 & 0 & 0 & 10 & 3 \\
 0 & 0 & 0 & 0 & 1 & 0 & 0 & 1 & 3 \\
 0 & 0 & 0 & 0 & 0 & 1 & 0 & 12 & 2 \\
 0 & 0 & 0 & 0 & 0 & 0 & 1 & 10 & 3
\end{pmatrix} \in \F_{13}^{7 \times 9}.
\end{align*}
The code $\mC+\mC'$ is not recovery balanced, which again can be checked with a computer algebra program. This shows that in general, the property of being recovery balanced is not preserved under the operation of summing.
\end{example}

\section{Breaking the Balance of Codes}\label{sec:break}

In this section we focus on code constructions for achieving random access expectation below~$k$. 
Inspired by~\cite[Construction 2]{bar2023cover} and also by Theorem~\ref{thm:sumra} and Lemma~\ref{lem:fi}, we show how we can ``perturb''
the recovery-balancedness of codes to obtain a better performance for the random access problem.
We do this by taking the generator matrix of a recovery balanced code and appending identity matrices to it. While we are not able to identify the optimal number of identity matrices one should append, we give closed formulas for the random access expectation and include plots for two different families of recovery balanced codes. 
We start with the following notation.


\begin{notation} \label{not:extmds}
Let $G= (I_k \mid R) \in \F_q^{k \times n}$ be a systematic generator matrix of a code. For $x \ge 1$, we let $G^x = (I_k \mid I_k \mid \dots \mid I_k \mid R) \in \F_q^{k \times N}$ (where $N=xk+n-k$) be the matrix obtained by appending additional $x-1$ identity matrices to $G$. 
\end{notation}

By concatenating identity matrices to the generator matrix of a recovery balanced code, it is possible to ``break'' the balance and improve its performance in terms of the random access expectation. Note that doing this with any recovery balanced code will improve\footnote{In the special case where $G=I_k$, the random access expectation of $G^x$ for all $x\ge 1$ is the same and equals to~$k$.} the random access expectation; however, deriving explicit formulas for their expectation is generally difficult. Because of this, we provide closed formulas for the expectation of these perturbed codes only for the generator matrices of MDS codes and simplex codes.

\subsection{MDS codes}
The following result gives an explicit formula for the random access expectation of ``perturbed'' MDS codes. 

\begin{theorem} \label{prop:extmds}
Let $G= (I_k \mid R) \in \F_q^{k \times n}$ be a systematic generator matrix of an MDS code and let $N=xk+n-k$. We have
\begin{align*}
\tmax(G^x) = 1+\sum_{s=1}^{N-1}\frac{\binom{N-x}{s}}{\binom{N-1}{s}} -\sum_{s=k}^{N-1}\sum_{a=0}^{k-1}\frac{\binom{k-1}{a}}{\binom{N-1}{s}} \sum_{m=0}^{s-k}\binom{n-k}{s-a-m}\sum_{t=0}^a(-1)^t\binom{a}{t}\binom{(a-t)x}{m+a}.
\end{align*}
\end{theorem}
\begin{proof}
In order to give an expression for $\E[\tau_i(G^x)]$, by (the natural analogue of) Lemma~\ref{lem:fi} it suffices to compute $\widetilde{\alpha}_i(s)$
for the code generated by $G^x$ and for $1\leq i\leq k$. Since $G$ is a systematic generator matrix of an MDS code, the only possible ways to recover $e_i$ is by either sampling one of the columns corresponding to $e_i$ itself, or by sampling at least~$k$ distinct columns of $G^x$. If $1\le s \le k-1$, it is not hard to see that 
\begin{align*}
        \widetilde{\alpha}_i(s) = \binom{N}{s}-\binom{N-x}{s}.
\end{align*}
We then focus on the case
$k\le s \le N-1$.
We can recover $e_i$ from $s$ columns of $G^x$ if~$e_i$ is one of them. The only other way to recover $e_i$ from~$s$ columns of $G^x$ is by sampling at least~$k$ distinct columns. Therefore, we need to find the size of the set
    $\{S \subseteq [N] : |S|=s, \, e_i \notin \{g_j : j \in S\}, \, |\{ g_j : j \in S\}| \ge k\}$,
    which we denote by $\mS$ in the sequel.
    
    Let $\mA=[xk]\setminus\{i,2i,\dots,xi\}$ and $\mB=\{xk+1,\dots,N\}$, so that we have $g_j=e_j$ for all $j \in \mA$ and the columns indexed by $\mB$ correspond to the columns of $R$, which is the redundancy part of $G=(I_k \mid R)$ as in Notation~\ref{not:extmds}.
    We have that $|\mS|$ is
    \begin{align} \label{eq:abab}
         \sum_{a=0}^{k-1} & |\{(A,B) \in 2^{\mA} \times 2^{\mB}  : |\{ g_j : j \in A\}|=a, |B| = s-|A|\}| \nonumber \\
        &=\sum_{a=0}^{k-1}\sum_{m=0}^{s-k}\sum_{\substack{A \subseteq \mA \\ |A|=a+m \\ |\{ g_j : j \in A\}|=a}}|\{B \subseteq \mB :   |B| = s-|A|\}| \nonumber \\
        &= \sum_{a=0}^{k-1}\sum_{m=0}^{s-k}\sum_{\substack{A \subseteq \mA \\ |A|=a+m \\ |\{ g_j : j \in A\}|=a}} \binom{N-xk}{s-a-m}.
    \end{align}
    For any $a \in \{0, \dots, k-1\}$ and $m \in \{0,\dots,s-k\}$, we are left with computing the size of the set
    \begin{align*}
        |\{A \subseteq \mA : |\{ g_j : j \in A\}|=a, |A|=a+m\}|.
    \end{align*}
    Let $C \subseteq \{g_j : j \in \mA\}$ with $|C|=a$. Using the Inclusion-Exclusion Principle, one can verify that we have
    \begin{align*}
        &|\{A \subseteq \mA : \{ g_j : j \in A\}=C, |A|=a+m\}| \\
        &= \sum_{t=0}^a(-1)^t\binom{a}{t}\binom{(a-t)x}{m+a}.
    \end{align*}
    Summing over all $C \subseteq \{g_j : j \in \mA\}$ with $|C|=a$ together with~\eqref{eq:abab} gives that $|\mS|$ equals
    \begin{align*}
      \sum_{a=0}^{k-1}\binom{k-1}{a} \sum_{m=0}^{s-k}\binom{N-xk}{s-a-m}\sum_{t=0}^a(-1)^t\binom{a}{t}\binom{(a-t)x}{m+a}.
    \end{align*}
By combining all of this
with Lemma~\ref{lem:fi} we finally obtain
\begin{align*}
\E[\tau_i(G^x)] = NH_N-\sum_{s=1}^{N-1}\frac{\binom{N}{s}-\binom{N-x}{s}}{\binom{N-1}{s}}-\sum_{s=k}^{N-1}\frac{|\mS|}{\binom{N-1}{s}},
\end{align*}
which does not depend on the coordinate $i$ and therefore, after simplifying, gives the statement of the proposition.
\end{proof}

Unfortunately, the formula in Theorem~\ref{prop:extmds} does not appear to be easy to evaluate explicitly. However, experimental results indicate that already for $x=2$ (and any MDS code) the random access expectation is strictly smaller than $k$. In Figure~\ref{fig:compi} we give an example of how the  expectation (normalized by~$k$) of the code obtained by concatenating $x-1$ identity matrices changes, depending on the value of $x$.

\begin{figure}[ht!]
\centering
\begin{tikzpicture}[scale=1]
\begin{axis}[legend style={at={(0.14,1)}, legend style={cells={align=left}}, anchor = north west, /tikz/column 2/.style={
                column sep=5pt}},
		legend cell align={left},
		width=15cm,height=9cm,
    xlabel={Values of $x$},
    xmin=1, xmax=40,
    ymin=0.87, ymax=1,
    xtick={1,5,10,15,20, 25,30,35,40},
    ytick={0.87,0.88,0.90,0.93,0.95,0.98,1},
    ymajorgrids=true,
    grid style=dashed,
     every axis plot/.append style={thick},  yticklabel style={/pgf/number format/fixed}
]
\addplot+[color=blue,mark=o,mark size=1pt,smooth]
coordinates {
(1,1.00000000000000000000000000000)
(2,0.933233433233433233433233433233)
(3,0.918296100184645076285943159007)
(4,0.915570407965077503373370919473)
(5,0.916881018277688363697990463555)
(6,0.919678919913444491762543362929)
(7,0.922991889036031650632669308391)
(8,0.926407253922990399285300448378)
(9,0.929741836208971959503781503563)
(10,0.932915696478146469112579119548)
(11,0.935898062150490532310023612070)
(12,0.938682249300216197980653338709)
(13,0.941273360572882817344580308240)
(14,0.943682017999798477162687416187)
(15,0.945921100146755077442832547453)
(16,0.948004034856083357225880483119)
(17,0.949943920549340946103916820008)
(18,0.951753096943807086308103845457)
(19,0.953442961423503093318476812609)
(20,0.955023919040230092159924537335)
(21,0.956505403566847504292497656537)
(22,0.957895934334727709184827966149)
(23,0.959203188962970152770601220256)
(24,0.960434080868200700668648691426)
(25,0.961594835505626422833247685096)
(26,0.962691062220706418611148576102)
(27,0.963727820280027804169630595195)
(28,0.964709678612279093949568085939)
(29,0.965640769323929522054826471767)
(30,0.966524835335110017383617326448)
(31,0.967365272613745297849437686393)
(32,0.968165167532715544493391826000)
(33,0.968927329873222621445159877027)
(34,0.969654321970053525568957677286)
(35,0.970348484454426583583833029102)
(36,0.971011959005399803605485944777)
(37,0.971646708475880578678899019801)
(38,0.972254534716528813638631796656)
(39,0.972837094381493111255445397752)
(40,0.973395912964437445201120220970)
    };
\addplot+[color=red,mark=o,mark size=1pt,smooth]
coordinates {
(1,1.00000000000000000000000000000)
(2,0.935978709972518022053625768796)
(3,0.913337503793823621525869473103)
(4,0.904382811826148981100916897724)
(5,0.901194053261328322489442718303)
(6,0.900765327460623575674010109245)
(7,0.901754317762166439881667487425)
(8,0.903494429080778780515176506302)
(9,0.905629361625693966766165177914)
(10,0.907958561955486445230272074201)
(11,0.910365094362792450246477180847)
(12,0.912779317185738807421843266335)
(13,0.915159447235973381633512654502)
(14,0.917480635793917135767666786204)
(15,0.919728575468554071111831578745)
(16,0.921895631925200416977749440877)
(17,0.923978438195513070351426188096)
(18,0.925976364624881968247936542518)
(19,0.927890528033247509060100976067)
(20,0.929723141010446487973244160006)
(21,0.931477080216684432806120720388)
(22,0.933155598170507970966314888314)
(23,0.934762130428230063347311526290)
(24,0.936300166945720808401412465891)
(25,0.937773167038028204418446297641)
(26,0.939184504165666916187191365655)
(27,0.940537431220794269769907303635)
(28,0.941835059929944356519789730539)
(29,0.943080349965942236714624710711)
(30,0.944276104704109178055021976970)
(31,0.945424971479706782424904533060)
(32,0.946529444842479186263729266022)
(33,0.947591871750562721886465610526)
(34,0.948614457960095712622707404064)
(35,0.949599275089079050343086892436)
(36,0.950548267991982301645484259854)
(37,0.951463262194219373581439323119)
(38,0.952345971216089472850636467283)
(39,0.953198003673300822495831058615)
(40,0.954020870082259119994107713752)
};
\addplot+[color=orange,mark=o,mark size=1pt,smooth]
coordinates {
(1,1.00000000000000000000000000000)
(2,0.942715426091385170669058137088)
(3,0.917267806323278587146653113669)
(4,0.904714402826546881853693643045)
(5,0.898385265696037554690271675431)
(6,0.895394441295338053545577185151)
(7,0.894320339738261249170450342988)
(8,0.894397283607533468792410702827)
(9,0.895185093988202170626699969043)
(10,0.896417904113350859446327837725)
(11,0.897928873873512564378295750171)
(12,0.899610140495966703926511507925)
(13,0.901390348669499080410583563683)
(14,0.903221476751031282958706948798)
(15,0.905070821692452348549606491589)
(16,0.906915965876868717931910464126)
(17,0.908741529051696767352666570069)
(18,0.910537021965000853810341547645)
(19,0.912295398383381926551631116940)
(20,0.914012060431759312810364319972)
(21,0.915684164451165011245552666709)
(22,0.917310129860674395530651356331)
(23,0.918889287480601310703774135310)
(24,0.920421625121213419905201364852)
(25,0.921907601932591657969202438114)
(26,0.923348011957761913301885774453)
(27,0.924743883277911970953047789683)
(28,0.926096403153671507499739917226)
(29,0.927406862316818032822004760972)
(30,0.928676613475978604358266600429)
(31,0.929907040441532393651375149917)
(32,0.931099535228437652482662894727)
(33,0.932255481180503293816358261254)
(34,0.933376240656201561896264308614)
(35,0.934463146179431776997350694927)
(36,0.935517494226654773940471334704)
(37,0.936540541021029070935017432258)
(38,0.937533499853291589622492555856)
(39,0.938497539561456212745734920994)
(40,0.939433783886524670981602760450)
};
\addplot+[color=green,mark=o,mark size=1pt,smooth]
coordinates {
(1,1.00000000000000000000000000000)
(2,0.949012207183121725850361532520)
(3,0.922899592829523499861954922850)
(4,0.908342788798196767836236526939)
(5,0.899915930437112407269766280028)
(6,0.895033537302291740856553601328)
(7,0.892326401234084889255104550285)
(8,0.891012336710355496247790724197)
(9,0.890618320669405758104464620157)
(10,0.890845345480105301426765216826)
(11,0.891497613471532199564837204031)
(12,0.892443188892723587774196649507)
(13,0.893591058124136638360835881248)
(14,0.894877213044396853452338813224)
(15,0.896255919648227795891166041466)
(16,0.897694079568526905127348633430)
(17,0.899167496488367059758108494490)
(18,0.900658348909613484124392787127)
(19,0.902153445864547986305702547159)
(20,0.903643001956158583638621659850)
(21,0.905119763652222604100772817220)
(22,0.906578377364226449606320117189)
(23,0.908014926628450283622352273686)
(24,0.909426589280899063045896809242)
(25,0.910811380912616540956168334373)
(26,0.912167961119785376948780843474)
(27,0.913495485966298598527888441606)
(28,0.914793494804217572557040059475)
(29,0.916061822879279608178989784646)
(30,0.917300533455338929598310145805)
(31,0.918509864832016098838507491711)
(32,0.919690188809132661646684724840)
(33,0.920841978007994844644536663563)
(34,0.921965780087583659375649877264)
(35,0.923062197358299970350564089844)
(36,0.924131870642515103451555995589)
(37,0.925175466491802357447040895418)
(38,0.926193667068158073856437429601)
(39,0.927187162147131048695428770266)
(40,0.928156642816427097465467971673)
};
\addplot+[color=purple,mark=o,mark size=1pt,smooth]
coordinates {
(1,1.00000000000000000000000000000)
(2,0.968267128362631839012355088718)
(3,0.946156521310231839091105895699)
(4,0.930241518018339728718174089934)
(5,0.918512762613550168372792294563)
(6,0.909721546279242171407680760493)
(7,0.903054907980618425307257394927)
(8,0.897963137082316800143856666904)
(9,0.894062695330802522490835966382)
(10,0.891078859033196965563488896200)
(11,0.888810408054843564732996210303)
(12,0.887107111491452649912513193247)
(13,0.885854931889842394986227632966)
(14,0.884966042807175785374860205753)
(15,0.884371936766199005305199097666)
(16,0.884018569011852544513017856608)
(17,0.883862873260899556155283587513)
(18,0.883870221077829136671399859436)
(19,0.884012542207037635114773656981)
(20,0.884266915546756384325442543390)
(21,0.884614500279364904171154427354)
(22,0.885039716202959260544202063241)
(23,0.885529608901134288829021009720)
(24,0.886073353573065463215395837884)
(25,0.886661863971030438743444223658)
(26,0.887287481780053438259848755261)
(27,0.887943728111549313659446189469)
(28,0.888625103355443894761705695088)
(29,0.889326924971101622847393713997)
(30,0.890045195256018362690654709322)
(31,0.890776492960670521497065129540)
(32,0.891517883991335485299317817431)
(33,0.892266847482398529901190421907)
(34,0.893021214312895366278609040041)
(35,0.893779115751687035621260075050)
(36,0.894538940387458968099764663907)
(37,0.895299297867233305742085249141)
(38,0.896058988255104884079586348131)
(39,0.896816976049961270744654583809)
(40,0.897572368080922736080897770025)
};
\addplot+[color=teal,mark=o,mark size=1pt,smooth, solid]
coordinates {
(1,1.00000000000000000000000000000)
(2,0.982221498921869966351642259494)
(3,0.967510945774038331465565721857)
(4,0.955215999017497626246411403665)
(5,0.944851528770537957660916669473)
(6,0.936050159494635275571741900302)
(7,0.928529154650707600291353364122)
(8,0.922067747522554695588519393506)
(9,0.916491306990092728301299382459)
(10,0.911660072218042514626869703856)
(11,0.907461002199844910766986472344)
(12,0.903801787643286835951385037845)
(13,0.900606389188033236255529960172)
(14,0.897811669693215638188422966051)
(15,0.895364821948346975222314383917)
(16,0.893221382309219737319925340351)
(17,0.891343681205009234765604562691)
(18,0.889699623065027736693971510158)
(19,0.888261717250959407322716783027)
(20,0.887006302114715717492540970138)
(21,0.885912919002004843135398396425)
(22,0.884963803666225082954425864384)
(23,0.884143470348378848125949512792)
(24,0.883438369539006598127123649872)
(25,0.882836604737501681955501320845)
(26,0.882327696761827865499521163933)
(27,0.881902386620353656293450878347)
(28,0.881552469839505791223945241992)
(29,0.881270656592320641970429841706)
(30,0.881050453100218559943467588101)
(31,0.880886060661723344744512963833)
(32,0.880772289355417153785254509020)
(33,0.880704484013497538609627224357)
(34,0.880678460499497038051661905664)
(35,0.880690450673745979320206504909)
(36,0.880737054711854028231492563168)
(37,0.880815199669324744775420265408)
(38,0.880922103370575049959066337761)
(39,0.881055242851789265776847595357)
(40,0.881212326710976555153456425081)
};
\legend{\tiny{rate $=1/2$}, \tiny{rate $=1/3$}, \tiny{rate $=1/4$}, \tiny{rate $=1/5$}, \tiny{rate $=1/10$}, \tiny{rate $=1/20$}
}
\end{axis}
\end{tikzpicture}
\caption{\label{fig:compi} Normalized random access coverage depth $\tmax(G^x)$ from Proposition~\ref{prop:fimds} for $k=5$ and various rates (where the rate is the ratio between the dimension $k=5$ and the $n$, the length of the MDS code that we started with).}
\end{figure}
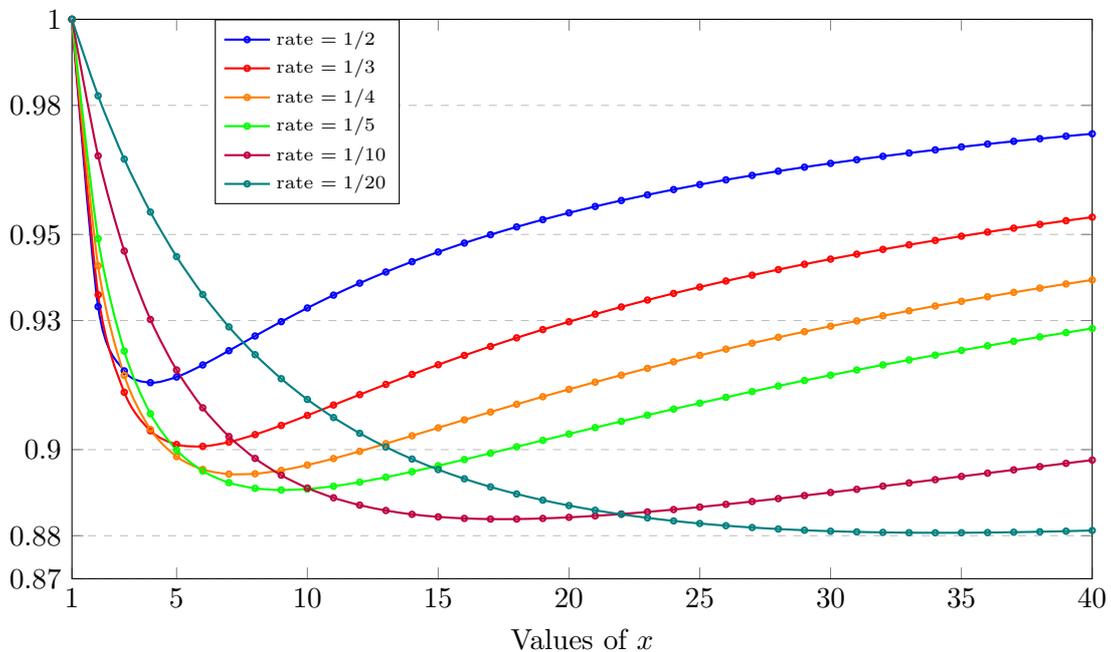

\begin{remark}
It is interesting to observe that for the case of the simple parity code, from experiments we believe that the optimal value of $x$ is 3. Even though the expression in Theorem~\ref{prop:extmds} simplifies substantially for that case, it is still not easy to formally prove that for any $k$ the optimal is attained for $x=3$. Moreover, whenever we consider MDS codes with redundancy larger that 1, there does not seem to be an optimal number for $x$, i.e., the lowest random access expecation is achieved for different values of $x$ (where we fix $n-k$ and consider different values of $n$ and $k$).
\end{remark}

\subsection{Simplex codes} \label{sec:extsim}

Let $G$ be the generator matrix of the $k$-dimensional simplex code in $\F_q^n$, and let $G^x$ be the same matrix with added $x-1$ copies of the identity $k \times k$ matrix. Note that for $x=1$ we have $G^x=G$. Then, \smash{$G^x \in \F_q^{k \times N}$}, with $N=n+(x-1)k$ and $n=(q^k-1)/(q-1)$. Recall that the columns of $G$ are all the projective points of $\F_q^k$.

\begin{theorem} \label{thm:extsim}
Let $G= (I_k \mid R) \in \F_q^{k \times n}$ be a systematic generator matrix of the 
$k$-dimensional simplex code in $\F_q^n$, $n=(q^k-1)/(q-1)$.
Let $N=xk+n-k$. Then, we have
\begin{multline*}
    \hspace{-2ex}T_{\textnormal{max}}(G^x) \hspace{-0.5ex}=\hspace{-0.5ex} NH_N\hspace{-0.25ex} \hspace{-0.25ex} \hspace{-0.25ex}-\hspace{-0.25ex}\sum_{s=1}^{N-1}\sum_{z=0}^k \sum_{\omega=0}^k\hspace{-0.25ex}
\binom{\frac{q^z-1}{q-1}\hspace{-0.25ex}+\hspace{-0.25ex}(x-1)\omega}{s}\hspace{-0.25ex} \left[ \eta(z,\omega) v_i^2(z) + \eta_i(z,\omega) \left( v_i^1(z) - v_i^2(z)\right)\right]\hspace{-0.25ex},
\end{multline*}
where:
\begin{align*}
    \eta(z,\omega) &=\binom{k}{\omega}\sum_{r=\omega}^k (-1)^{r-\omega} \binom{k-\omega}{r-\omega}\qbin{k-r}{z-r}{q}, \\
    \eta_i(z,\omega) &= \hspace{-0.25ex}\binom{k-1}{\omega-1}\sum_{r=\omega}^k \binom{k-\omega}{r-\omega}\qbin{k-r}{z-r}{q}(-1)^{r-\omega}  +\hspace{-0.25ex}\binom{k-1}{\omega}\hspace{-0.25ex}\sum_{r=\omega+1}^k \hspace{-0.25ex}\binom{k-\omega-1}{r-\omega-1}\qbin{k-r}{z-r}{q}(-1)^{r-\omega} \\
&\qquad + 
\binom{k-1}{\omega}\sum_{r=\omega}^{k-1} \binom{k-1-\omega}{r-\omega}\qbin{k-r-1}{z-r-1}{q}(-1)^{r-\omega},    \\
v_i^1(z) &= \sum_{h=z}^k (-1)^{h-z} q^{\binom{h-z}{2}} \qbin{k-z}{h-z}{q}, \quad \quad  v_i^2(z) = \sum_{h=z+1}^k (-1)^{h-z} q^{\binom{h-z}{2}} \qbin{k-z-1}{h-z-1}{q}.
\end{align*}
\end{theorem}

In the remainder of this subsection, we give a proof for Theorem~\ref{thm:extsim}. Note that it suffices to give a formula for 
\begin{align*}
\alpha_i(s) := {|\{S \subseteq [N] : |S| = s, \, e_i \in \langle g_j : j \in S\rangle \}|}
\end{align*}
for all $i \in [k]$ and all $s \in \{1, \ldots, N-1\}$
and then apply Lemma~\ref{lem:fi}.

We split the proof into four parts, the first three of which are 
independent lemmas. The three lemmas and their proofs can be found in Appendix~\ref{sec:apen}.
In the sequel (and also in Appendix~\ref{sec:apen}), for a subspace $V \le \F_q^k$ we let $I(V)=\{t : e_t \in V\}$ and $\omega(V)=|I(V)|$.

\begin{proof}[Proof of Theorem~\ref{thm:extsim}]
We fix $i$ and $s$ throughout the proof.  For a subspace $V \le \F_q^k$, let
 $f(V)=|\{S \subseteq [N] \, : \, |S|=s, \, \langle g_j : j \in S \rangle =V\}|$, where the $g_j$'s are columns of $G^x$.
 Note that 
 \begin{equation} \label{aha}
     \alpha_i(s)=\sum_{\substack{V \le \F_q^k \\ e_i \in V}}f(V).
 \end{equation}
 Define $g(V)=\sum_{U \le V}f(V)$ and note that
 $$g(V) = |\{S \subseteq [N] \, : \, |S|=s, \, g_j \in V \mbox{ for all $j \in S$}\}|.$$
By letting $\omega=\omega(V)$ and $z=\dim(V)$, we then have
 \begin{equation}
     g(V)= \binom{\frac{q^z-1}{q-1}+(x-1)\omega}{s}.
 \end{equation}
 This follows from the fact that $G^x$ has as columns all the projective points of $\F_q^k$, in addition to $x-1$ copies of the identity matrix.
We now have, by Möbius Inversion in the lattice of subspaces of $\F_q^k$,
 $$\alpha_i(s)=\sum_{\substack{V \le \F_q^n \\ e_i \in V}} \sum_{U \le V} \mu(U,V) g(U),$$
 where $\mu$ is the Möbius function of the lattice.
 We exchange the order of summation and split over dimensions and value of $\omega(U)$, obtaining 
 \begin{align*}
     \alpha_i(s) &= \sum_{z=0}^k \sum_{\omega=0}^k \sum_{\substack{U \le \F_q^k \\ \dim(U)=z\\ \omega(U)=\omega}} \binom{\frac{q^z-1}{q-1}+(x-1)\omega}{s} \sum_{h=z}^k (-1)^{h-z} q^{\binom{h-z}{2}} v_i(h,U).
 \end{align*}
 Lastly, we apply Lemmas~\ref{lemm1}, \ref{lemm2}, and \ref{lemm3} from Appendix~\ref{sec:apen} and obtain that
  \begin{align*}
     \alpha_i(s) &= \sum_{z=0}^k \sum_{\omega=0}^k \sum_{\substack{U \le \F_q^k \\ \dim(U)=z\\ \omega(U)=\omega \\ e_i \in U}} \binom{\frac{q^z-1}{q-1}+(x-1)\omega}{s} \sum_{h=z}^k (-1)^{h-z} q^{\binom{h-z}{2}}\qbin{k-z}{h-z}{q} \\
     &+ \sum_{z=0}^k \sum_{\omega=0}^k \sum_{\substack{U \le \F_q^k \\ \dim(U)=z\\ \omega(U)=\omega \\ e_i \notin U}} \binom{\frac{q^z-1}{q-1}+(x-1)\omega}{s} \sum_{h=z+1}^k (-1)^{h-z} q^{\binom{h-z}{2}} \qbin{k-z-1}{h-z-1}{q} \\
     &= \sum_{z=0}^k \sum_{\omega=0}^k \binom{\frac{q^z-1}{q-1}+(x-1)\omega}{s} \eta_i(z,\omega)\sum_{h=z}^k (-1)^{h-z} q^{\binom{h-z}{2}} \qbin{k-z}{h-z}{q} \\
     &+
     \sum_{z=0}^k \sum_{\omega=0}^k \binom{\frac{q^z-1}{q-1}+(x-1)\omega}{s} (\eta(z,\omega) - \eta_i(z,\omega))\sum_{h=z+1}^k (-1)^{h-z} q^{\binom{h-z}{2}}  \qbin{k-z-1}{h-z-1}{q}. 
 \end{align*}
Rearranging the terms one gets the desired expression. 
\end{proof}

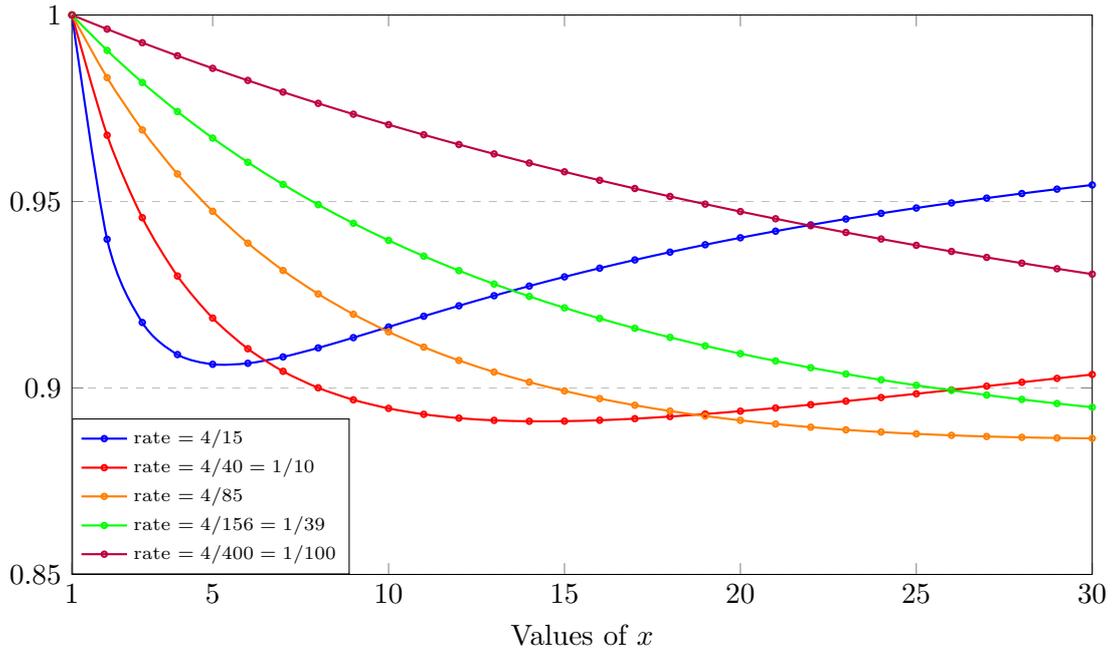
\begin{figure}[h]
\centering
\begin{tikzpicture}[scale=1]
\begin{axis}[legend style={at={(0.00,0.00)}, anchor = south west, /tikz/column 2/.style={
                column sep=5pt}},
		legend cell align={left},
		width=15cm,height=9cm,
    xlabel={Values of $x$},
    xmin=1, xmax=30,
    ymin=0.85, ymax=1,
    xtick={1,5,10,15,20,25,30},
    ytick={0.85,0.90,0.95,1},
    ymajorgrids=true,
    grid style=dashed,
     every axis plot/.append style={thick},  yticklabel style={/pgf/number format/fixed}
]
\addplot+[color=blue,mark=o,mark size=1pt,smooth]
coordinates {
(1,1.00000000000000000000000000000)
(2,0.939872729394788218317630082336)
(3,0.917559523809523809523809523809)
(4,0.908960588644092480409615959488)
(5,0.906353021978021978021978021978)
(6,0.906620697368476052663370769024)
(7,0.908325501253132832080200501253)
(8,0.910744595442871304940270457512)
(9,0.913493915742914598749839024439)
(10,0.916362281011671255573694598085)
(11,0.919231672932330827067669172932)
(12,0.922036328301414691827124170705)
(13,0.924740652492668621700879765396)
(14,0.927326807927157847863098840922)
(15,0.929787515384941855530090824208)
(16,0.932121774040610543062945796027)
(17,0.934332268373015401877880655470)
(18,0.936423775923356971750553076120)
(19,0.938402181861208382544598951642)
(20,0.940273867992024758711837707551)
(21,0.942045336062823955298243811059)
(22,0.943722979206358304032806821654)
(23,0.945312947684271213682978388861)
(24,0.946821074825019054936816102578)
(25,0.948252841312434335690149643638)
(26,0.949613363712292534872708475503)
(27,0.950907398063630778225392624933)
(28,0.952139352552386719528685739566)
(29,0.953313305365390663044086264707)
(30,0.954433025189403206767761559403)
    };
\addplot+[color=red,mark=o,mark size=1pt,smooth]
coordinates {
(1,1.00000000000000000000000000000)
(2,0.967775641632830975538703287337)
(3,0.945658081298154192386554585688)
(4,0.930022918258212375859434682964)
(5,0.918748640265784257664580686513)
(6,0.910518944283601975650713802585)
(7,0.904478253085064230575066488379)
(8,0.900049040014073639272420718208)
(9,0.896829121164238968376506629745)
(10,0.894531024531024531024531024531)
(11,0.892944695106907150006679467421)
(12,0.891913734189836397023192757838)
(13,0.891319795822838419187303568643)
(14,0.891072066334399053953435025417)
(15,0.891100005213040487475280843093)
(16,0.891348232999775649503417198517)
(17,0.891772865261838939990104795142)
(18,0.892338840607487081424602435330)
(19,0.893017944630847856654308267211)
(20,0.893787329192223049710215929800)
(21,0.894628389575342474938917793314)
(22,0.895525903754492114544513022462)
(23,0.896467366034566001480202497415)
(24,0.897442466484991019487091822316)
(25,0.898442680881705271949174388199)
(26,0.899460945224482475386870009754)
(27,0.900491395562893818941469924160)
(28,0.901529158672015814872957730101)
(29,0.902570182625192716418477904283)
(30,0.903611098894616155078436313143)
};
\addplot+[color=orange,mark=o,mark size=1pt,smooth]
coordinates {
(1,1.00000000000000000000000000000)
(2,0.983209883742583812233163349404)
(3,0.969189454897686623137167358900)
(4,0.957384761157243268717922121879)
(5,0.947375380878375897894002618730)
(6,0.938837219004969490153280669748)
(7,0.931516908481376525668529460056)
(8,0.925213853441890502481124074371)
(9,0.919767409314823107926556202418)
(10,0.915047581903932113872891083156)
(11,0.910948179251379771554608096890)
(12,0.907381700619400911747573977834)
(13,0.904275473987219676874849288642)
(14,0.901568703195720408819659810892)
(15,0.899210186269162080468595839574)
(16,0.897156534804725264554806452421)
(17,0.895370771542646542646542646543)
(18,0.893821216289782952040986772688)
(19,0.892480593807224375311390050735)
(20,0.891325314085353659883734900628)
(21,0.890334887624177553367167843957)
(22,0.889491447273820014914079280246)
(23,0.888779354804276229297228598965)
(24,0.888184875313977939800765664352)
(25,0.887695906309248027514281384251)
(26,0.887301751114752682365442981211)
(27,0.886992928440335339109581944580)
(28,0.886761011599734495755039483525)
(29,0.886598492173927212735949380130)
(30,0.886498663926168070321851609277)
};
\addplot+[color=green,mark=o,mark size=1pt,smooth]
coordinates {
(1,1.00000000000000000000000000000)
(2,0.990510951365027326519858557456)
(3,0.981915429684131862891617759613)
(4,0.974108564720504812699130003055)
(5,0.967000615138909395679346288928)
(6,0.960514405585399500206802438040)
(7,0.954583256571997003064709589191)
(8,0.949149302123062282879197130311)
(9,0.944162114550131804953583776712)
(10,0.939577574022639690216855283027)
(11,0.935356934425042436351872016542)
(12,0.931466047506409943047312152043)
(13,0.927874715382373595832978825429)
(14,0.924556147657797023474590104757)
(15,0.921486504254202461322918937555)
(16,0.918644508783514123084309748623)
(17,0.916011120258275418892564136102)
(18,0.913569253254953132136111656743)
(19,0.911303538492158021666230441679)
(20,0.909200117256514000384617287654)
(21,0.907246464287114875189389831603)
(22,0.905431234677667372893461651259)
(23,0.903744131121879382493282359400)
(24,0.902175788449917623875100406552)
(25,0.900717672911068201468567535225)
(26,0.899361994072991515446758413510)
(27,0.898101627549196151495842695113)
(28,0.896930047047815392240364867494)
(29,0.895841264467769731726701856127)
(30,0.894829776961960745947970968114)
};
\addplot+[color=purple,mark=o,mark size=1pt,smooth]
coordinates {
(1,1.00000000000000000000000000000)
(2,0.996219255314008710787248605047)
(3,0.992581919104139518297864676918)
(4,0.989081063855459191499379384557)
(5,0.985710179474116998170199196171)
(6,0.982463143396264807818105424165)
(7,0.979334193152545755352274705147)
(8,0.976317901162265644466726194399)
(9,0.973409151554239835421862285876)
(10,0.970603118831666331420848249895)
(11,0.967895248216512495958750958574)
(12,0.965281237525079594079376644297)
(13,0.962757020440854543695399965976)
(14,0.960318751063671785970993221241)
(15,0.957962789625765223795062632606)
(16,0.955685689275644464324738440805)
(17,0.953484183840016205603350569087)
(18,0.951355176482309350513567718392)
(19,0.949295729183855969309908226372)
(20,0.947303052980521852309488790247)
(21,0.945374498893651573492866263597)
(22,0.943507549499665738183801452848)
(23,0.941699811087586136581863479512)
(24,0.939949006358224318295149538197)
(25,0.938252967622800700207515056458)
(26,0.936609630462409161246659286946)
(27,0.935017027813045656460521023929)
(28,0.933473284443913833043683521338)
(29,0.931976611799437239631669591344)
(30,0.930525303177874376627516565532)
};


\legend{\tiny{rate $=4/15$}, \tiny{rate $=4/40=1/10$}, \tiny{rate $=4/85$}, \tiny{rate $=4/156=1/39$}, \tiny{rate $=4/400=1/100$}}
\end{axis}
\end{tikzpicture}
\caption{\label{fig:comparison} Normalized random access coverage depth $\tmax(G^x)$ from Theorem~\ref{thm:extsim} for $k=3$ and various rates (where the rate is the ratio between the dimension $k=3$ and $(q^3-1)/(q-1)$ for different values of the underlying field size $q$).}
\end{figure}

While the plots in Figure~\ref{fig:compi} and Figure~\ref{fig:comparison} both show that the random access expectation decreases when identity matrices are appended to the generator matrix of MDS codes and simplex codes, respectively, and although the plots look similar, they are not comparable in general. The formula from Theorem~\ref{prop:extmds} can be applied to codes of any rate; however, assuming the MDS conjecture holds, such codes can only be constructed when $n \leq q+1$. On the other hand, the formula in Theorem~\ref{thm:extsim} can be applied to any dimension and any $q$, but the length will be fully determined by the choice of dimension and $q$. We believe both of these formulas are interesting in their own right for this reason.

\section{Discussion and Future Work} \label{sec:conclusion}
In this paper, we studied the random access coverage depth problem and specifically, we focused on the task of determining the values of $\tmax(G)$, $T(n,k)$, and $T(k)$. Our results give important steps towards understanding what structural properties of generating matrices result in random access expectation that is smaller than $k$. While the results presented in this paper significantly contribute to the study of the random access coverage depth problem, several interesting questions remain open, as listed below.
\begin{itemize}
     \item A natural question inspired by our results is to better understand which codes are recovery balanced. Specifically: \begin{itemize} 
     \item We conjecture that the property of being recovery balanced is closed under duality and that there are only a few other families of recovery balanced codes beyond those presented here. 
     \item While we have presented several sufficient conditions for a code to be recovery balanced, finding a useful condition that is both necessary and sufficient remains unsolved. 
     \item Understanding if, and under which conditions, additional code operations preserve the property of being recovery balanced is also of great interest. 
    \end{itemize}
    \item In Remark~\ref{rem:avvg} we discussed the difference between systematic and non-systematic codes on the random access coverage depth. In particular, experiments and intuitive reasoning suggest that systematic generator matrices outperform non-systematic generator matrices for this problem. However, proving this property remains for future work. 
    \item In Section~\ref{sec:break} we demonstrated that using duplication of the systematic part (i.e., appending copies of the identity matrix to a generating matrix of a code) can disrupt the balance of a code in our favor and reduce the random access coverage depth expectation below $k$. While we have presented closed-form expressions for the expectation using this technique for several codes, a full characterization of performance using this technique remains unsolved. In particular, simulations suggest that there is an optimal number of identity matrices to append to minimize the expectation (e.g., three in the case of a simple parity code). However, proving such behavior and identifying this optimal number are both open problems that should be explored in future work. 
    \item Lastly, there remains an interesting gap between our best codes, in terms of minimizing the expected random access coverage depth, and the lower bounds presented in~\cite{bar2023cover}. Closing this gap is also an important task for future research. 
\end{itemize}

\section*{Acknowledgement}
The authors thank Itzhak Tamo for helpful discussions, and in particular, for pointing out that codes with a transitive permutation automorphism group are recovery balanced. 
\bibliographystyle{ieeetr}
\bibliography{ourbib}
\clearpage

\begin{appendix}
\section{}\label{sec:apen}
\begin{proof}[Statements and proofs for the formulas used in Remark~\ref{rem:avvg}] We first prove the general average and then the average over all matrices in systematic form. Note that the formulas we compute do not depend on the index $i$, so for both statements we choose and fix an arbitrary $i \in [k]$. Recall that we let $\mG:=\{G \in \F_q^{k \times n}: \rk(G)=k\}$ and $\mG_{\operatorname{sys}} :=\{G \in \F_q^{k \times n}: \rk(G)=k, \; \textnormal{$G$ is systematic}\}$. We use the formula of Lemma~\ref{lem:fi} in both statements, and when we write $X \le Y$ for subspaces $X$ and $Y$, we mean that $X$ is an $\F_q$-subspace of $Y$.

\begin{claim}[General Average] \label{cl:genav}
    We have
    \begin{align*}
        \frac{\sum_{G \in \mG}\E[{\tau}_i(G)]}{|\mG|} = &\sum_{s=0}^{n-1}\frac{\binom{n}{s}}{\prod_{j=0}^{k-1}(q^n-q^j)\binom{n-1}{s}}\sum_{u=0}^s \left(\qbin{k}{u}{q}-\qbin{k-1}{u-1}{q}\right)\phi(u,s) \\
        &\cdot \sum_{v=k-u}^k\sum_{l=u}^k\qbin{k-u}{l-u}{q}(-1)^{k-l}q^{\binom{k-l}{2}} \qbin{l}{v}{q}\phi(v,n-s),
    \end{align*}
    where 
    \begin{align*}
    \phi(\ell,t)= \sum_{d=0}^{\ell}\sum_{\substack{D \le U \\ \dim(D)=d}}(-1)^{{\ell}-d}q^{\binom{{\ell}-d}{2}}q^{dt} = \sum_{d=0}^{\ell}\qbin{{\ell}}{d}{q}(-1)^{{\ell}-d}q^{\binom{{\ell}-d}{2}}q^{dt}.
\end{align*}
\end{claim}
\begin{proof}
Denote by $g_j$ the $j$-th column of a matrix $G$.
We have
\begin{align*}
   \sum_{G \in \mS}\E[{\tau}_i(G)] &= \sum_{G \in \mG}\left(nH_n-\sum_{s=1}^{n-1}\frac{|\{S \subseteq [n] : |S|=s, \, e_i \in \langle g_j :j \in S \rangle\}|}{\binom{n-1}{s}}\right) \\
    &= |\mG| nH_n- \sum_{G \in \mG}\sum_{s=1}^{n-1}\frac{|\{S \subseteq [n] : |S|=s, \, e_i \in \langle g_j :j \in S \rangle\}|}{\binom{n-1}{s}}\\
    &= |\mG| nH_n- \sum_{s=1}^{n-1}\sum_{\substack{S \subseteq [n]\\ |S| = s}}\frac{|\{G \in \mG : e_i \in \langle g_j :j \in S \rangle\}|}{\binom{n-1}{s}}\\
    &= |\mG| nH_n- \sum_{s=1}^{n-1}\sum_{\substack{S \subseteq [n]\\ |S| = s}}\frac{|\mG|-|\{G \in \mG : e_i \notin \langle g_j :j \in S \rangle\}|}{\binom{n-1}{s}} \\
    &= |\mG| nH_n- \sum_{s=1}^{n-1}\frac{\binom{n}{s}|\mG|}{\binom{n-1}{s}}+\sum_{s=1}^{n-1}\sum_{\substack{S \subseteq [n]\\ |S| = s}}\frac{|\{G \in \mG : e_i \notin \langle g_j :j \in S \rangle\}|}{\binom{n-1}{s}}  \\
    &= |\mG| nH_n- |\mG| nH_{n-1}+\sum_{s=1}^{n-1}\sum_{\substack{S \subseteq [n]\\ |S| = s}}\frac{|\{G \in \mG : e_i \notin \langle g_j :j \in S \rangle\}|}{\binom{n-1}{s}} \\
    &= \sum_{s=0}^{n-1}\sum_{\substack{S \subseteq [n]\\ |S| = s}}\frac{|\{G \in \mG : e_i \notin \langle g_j :j \in S \rangle\}|}{\binom{n-1}{s}},
\end{align*}
where the second-to-last equality can be obtained similarly to what we did in the proof of Lemma~\ref{lem:fi}.
For a fixed set $S \subseteq [n]$ with $|S|=s$ let us denote by $\mS$ the set $\{G \in \mG : e_i \notin \langle g_j :j \in S \rangle\}$. For a matrix $M \in \F_q^{t \times s}$ we denote by $\textnormal{colsp}(M)$ the column-space of $M$, i.e., $\textnormal{colsp}(M):=\{Mv : v \in \F_q^s\} \le \F_q^t$. 
We have 
\begin{align*}
    |\mS| &= |\{G \in \mG : e_i \notin \langle g_j :j \in S \rangle\}| \\
    &=  \sum_{\substack{M \in \F_q^{k \times s} \\ e_i \notin \textnormal{colsp}(M)}}|\{N \in \F_q^{k \times (n-s)} : \textnormal{colsp}(M)+\textnormal{colsp}(N)=\F_q^k\}| \\
    &= \sum_{u=0}^s \sum_{\substack{U \le \F_q^k\\ \dim(U)=u \\ e_i \notin U}}\sum_{\substack{M \in \F_q^{k \times s} \\ \textnormal{colsp}(M)=U}}|\{N \in \F_q^{k \times (n-s)} : U+\textnormal{colsp}(N)=\F_q^k\}| \\
    &=\sum_{u=0}^s \sum_{\substack{U \le \F_q^k\\ \dim(U)=u \\ e_i \notin U}}\sum_{\substack{M \in \F_q^{k \times s} \\ \textnormal{colsp}(M)=U}}\sum_{\substack{V \le \F_q^k \\ U+V=\F_q^k}}|\{N \in \F_q^{k \times (n-s)} : \textnormal{colsp}(N)=V\}| \\
    &=\sum_{u=0}^s \sum_{\substack{U \le \F_q^k\\ \dim(U)=u \\ e_i \notin U}}\sum_{\substack{M \in \F_q^{k \times s} \\ \textnormal{colsp}(M)=U}}\sum_{v=k-u}^k\sum_{\substack{V \le \F_q^k \\ U+V=\F_q^k \\ \dim(V)=v}}|\{N \in \F_q^{k \times (n-s)} : \textnormal{colsp}(N)=V\}|.
\end{align*}
To evaluate the above formula we need to count $|\{V \le \F_q^k : \dim(V)=v, U+V=\F_q^k\}|$ for some fixed $U \le \F_q^k$ of dimension $u$ and for $v \in \{k-u,\dots,k\}$. Let
\begin{align*}
    f(W)&=|\{V \le \F_q^k : \dim(V)=v, U+V=W\}|, \\
    g(W)&=\sum_{L \le W}f(L)=|\{V \le \F_q^k : \dim(V)=v, U+V\le W\}| \\
    &=|\{V \le \F_q^k : \dim(V)=v, U\le W, V \le W\}| \\
    &=\begin{cases}
        0 \quad \quad \quad \quad \quad \textnormal{if $U \nleq W$}, \\
        \qbin{\dim(W)}{v}{q} \quad \textnormal{otherwise.}
    \end{cases}
\end{align*}
This implies that 
$$f(W)=\sum_{U \le L \le W} (-1)^{\dim(W)-\dim(L)}q^{\binom{\dim(W)-\dim(L)}{2}} \qbin{\dim(L)}{v}{q}.$$
Therefore, by Möbius Inversion, we have
\begin{align*}
    f(\F_q^k)=|\{V \le \F_q^k : \dim(V)=v, U+V=\F_q^k\}| &= \sum_{l=u}^k\sum_{\substack{U \le L \le \F_q^k\\ \dim(L)=l}}(-1)^{k-l}q^{\binom{k-l}{2}} \qbin{l}{v}{q} \\
    &= \sum_{l=u}^k\qbin{k-u}{l-u}{q}(-1)^{k-l}q^{\binom{k-l}{2}} \qbin{l}{v}{q}.
\end{align*}
For a subspace $U \le \F_q^k$ we let 
\begin{align*}
    f(U)&=|\{M \in \F_q^{k \times (t-a)} :\textnormal{colsp}(M) = U\}|, \\
    g(U)&=\sum_{W \le U} f(W) \\
    &= |\{M \in \F_q^{k \times (t-a)} :\textnormal{colsp}(M) \le U\}|  \\
    &= q^{u(t-a)}.
\end{align*}
Then by Möbius Inversion we get
\begin{align*}
    \phi(u,t-a)=f(U) = \sum_{d=0}^u\sum_{\substack{D \le U \\ \dim(D)=d}}(-1)^{u-d}q^{\binom{u-d}{2}}q^{d(t-a)} = \sum_{d=0}^u\qbin{u}{d}{q}(-1)^{u-d}q^{\binom{u-d}{2}}q^{d(t-a)}.
\end{align*}
Therefore, we also have
\begin{align*}
    &|\{M \in \F_q^{k \times s} : \textnormal{colsp}(M)=U\}| = \sum_{d=0}^u\qbin{u}{d}{q}(-1)^{u-d}q^{\binom{u-d}{2}}q^{ds}, \\
    &|\{N \in \F_q^{k \times (n-s)} : \textnormal{colsp}(N)=V\}| = \sum_{r=0}^v\qbin{v}{r}{q}(-1)^{v-r}q^{\binom{v-r}{2}}q^{r(n-s)}.
\end{align*}
Finally, we have that $\sum_{G \in \mS}\E[{\tau}_i(G)]$ equals
\begin{align*}
    \sum_{G \in \mS}\E[{\tau}_i(G)] &= \sum_{s=0}^{n-1}\sum_{\substack{S \subseteq [n]\\ |S| = s}}|\mS|\\ &=\sum_{s=0}^{n-1}\sum_{\substack{S \subseteq [n]\\ |S| = s}}\sum_{u=0}^s \sum_{\substack{U \le \F_q^k\\ \dim(U)=u \\ e_i \notin U}}\sum_{\substack{M \in \F_q^{k \times s} \\ \textnormal{colsp}(M)=U}}\sum_{v=k-u}^k\sum_{\substack{V \le \F_q^k \\ U+V=\F_q^k \\ \dim(V)=v}}|\{N \in \F_q^{k \times (n-s)} : \textnormal{colsp}(N)=V\}| \\
    &=\sum_{s=0}^{n-1}\frac{\binom{n}{s}}{\binom{n-1}{s}}\sum_{u=0}^s \left(\qbin{k}{u}{q}-\qbin{k-1}{u-1}{q}\right)\sum_{d=0}^u\qbin{u}{d}{q}(-1)^{u-d}q^{\binom{u-d}{2}}q^{ds}\cdot \\
    &\sum_{v=k-u}^k\sum_{l=u}^k\qbin{k-u}{l-u}{q}(-1)^{k-l}q^{\binom{k-l}{2}} \qbin{l}{v}{q}\sum_{r=0}^v\qbin{v}{r}{q}(-1)^{v-r}q^{\binom{v-r}{2}}q^{r(n-s)}
\end{align*}
and dividing by the number of $k \times n$ matrices over $\F_q$, i.e., by $\prod_{j=0}^{k-1}(q^n-q^j)$, gives the average random access expectation over all generator matrices in the statement of the claim. 
\end{proof}

\begin{claim}[Systematic Average]
    We have
    \begin{align*}
        \frac{\sum_{G \in \mG_{\operatorname{sys}}}\E[{\tau}_i(G)]}{|\mG_{\operatorname{sys}}|} &= \sum_{s=0}^{n-1} \frac{1}{q^{k(n-k)}\binom{n-1}{s}} \sum_{a=0}^{k-1} \binom{k-1}{a} 
\binom{n-k}{s-a} q^{k(n-k-s+a)}\\ &\cdot \sum_{u=0}^{k-1} \sum_{v=0}^{k-1}\left(  \qbin{k-a}{v-a}{q}-\qbin{k-a-1}{v-a-1}{q}\right)\sum_{w=a}^v \qbin{v-a}{w-a}{q} (-1)^{v-w} q^{\binom{v-w}{2}}\qbin{w}{u}{q}\phi(u,s-a).
    \end{align*}
    where $ \phi(\ell,t)$ is defined as in Claim~\ref{cl:genav}.
\end{claim}
\begin{proof}
As in the proof of Claim~\ref{cl:genav} we can write
\begin{align*}
    \sum_{G \in \mG_{\operatorname{sys}}}\E[{\tau}_i(G)] &= \sum_{s=0}^{n}\sum_{\substack{S \subseteq [n]\\ |S| = s}}\frac{|\{G \in \mG_{\operatorname{sys}} : e_i \notin \langle g_j :j \in S \rangle\}|}{\binom{n-1}{s}} \\
    &= \sum_{s=0}^{n}\sum_{\substack{S \subseteq [n]\setminus \{i\}\\ |S| = s}}\frac{|\{G \in \mG_{\operatorname{sys}} : e_i \notin \langle g_j :j \in S \rangle\}|}{\binom{n-1}{s}} \\
    &= \sum_{s=0}^{n}\sum_{a=0}^{k-1} \sum_{\substack{A \subseteq [k] \setminus \{i\} \\ |A|=a}} 
\sum_{\substack{B \subseteq \{k+1, \ldots, n\} \\ |B|=s-a}} q^{k(n-k-s+a)}|\{M \in \F_q^{k \times (s-a)} :e_i \notin \mbox{colsp}(M) + \langle e_j :j \in A \rangle\}| \\
&= \sum_{s=0}^{n}\sum_{a=0}^{k-1} \binom{k-1}{a} 
\binom{n-k}{s-a} q^{k(n-k-s+a)}|\{M \in \F_q^{k \times (s-a)} :e_i \notin \mbox{colsp}(M) + \langle e_j :j \in A \rangle\}|.
\end{align*}
In order to compute $|\{M \in \F_q^{k \times (s-a)} :e_i \notin \mbox{colsp}(M) + \langle e_j :j \in A \rangle\}|$ we write
\begin{align*}
    &|\{M \in \F_q^{k \times (s-a)} :e_i \notin \mbox{colsp}(M) + \langle e_j :j \in A \rangle\}| \\
    &= \sum_{\substack{U \le \F_q^k \\ e_i \notin U}}|\{M \in \F_q^{k \times (s-a)} : \mbox{colsp}(M) + \langle e_j :j \in A \rangle = U\}| \\
    &= \sum_{u=0}^{k-1}\sum_{\substack{U \le \F_q^k \\ \dim(U)=u \\ e_i \notin U+\langle e_j :j \in A \rangle}}|\{M \in \F_q^{k \times (s-a)} : \mbox{colsp}(M)  = U\}| \\
    &= \sum_{u=0}^{k-1}\sum_{\substack{U \le \F_q^k \\ \dim(U)=u \\ e_i \notin U+\langle e_j :j \in A \rangle}}\phi(u,s-a) \\
    &= \sum_{u=0}^{k-1} |\{U \le \F_q^k :\dim(U)=u, e_i \notin U+\langle e_j :j \in A \rangle\}|\phi(u,s-a) 
\end{align*}
where $\phi(u,s-a)= |\{M \in \F_q^{k \times (s-a)} : \mbox{colsp}(M)  = U\}| $ for $U \le \F_q^k$ with $\dim(U) = u$ as in the statement and proof of Claim~\ref{cl:genav}. In order to compute $|\{U \le \F_q^k :\dim(U)=u, e_i \notin U+\langle e_j :j \in A \rangle\}|$ we write
\begin{multline*}
    |\{U \le \F_q^k :\dim(U)=u, e_i \notin U+\langle e_j :j \in A \rangle\}| = \\  \sum_{v=0}^{k-1}\sum_{\substack{V \le \F_q^k \\ e_i \notin V \\ \dim(V) = v}}|\{U \le \F_q^k :\dim(U)=u, U+\langle e_j :j \in A \rangle =V \}|
\end{multline*}
Then we use Möbius Inversion
\begin{align*}
    f(V)&= |\{U \le \F_q^k :\dim(U) = u, U+\langle e_j :j \in A \rangle = V \}|, \\
    g(V)&= \sum_{W \le V} f(W) \\
    &= |\{U \le \F_q^k :\dim(U) = u, U+\langle e_j :j \in A \rangle \le V \}| \\
    &= |\{U \le \F_q^k :\dim(U) = u, U \le V, \langle e_j :j \in A \rangle \le V \}| \\
    &= \begin{cases}
        0 \quad &\textnormal{if $\langle e_j :j \in A \rangle \nleq V$,} \\
        \qbin{v}{u}{q} \quad &\textnormal{otherwise.}
    \end{cases}
\end{align*}
We have
\begin{align*}
    f(V)= \sum_{W \le V} (-1)^{v-w} q^{\binom{v-w}{2}}g(W)=\sum_{\langle e_j :j \in A \rangle \le W \le V} (-1)^{v-w} q^{\binom{v-w}{2}}\qbin{w}{u}{q}.
\end{align*}
Therefore we obtain
\begin{multline*}
    |\{U \le \F_q^k :\dim(U)=u, e_i \notin U+\langle e_j :j \in A \rangle\}| = \\  \sum_{v=0}^{k-1}\sum_{\substack{V \le \F_q^k \\ e_i \notin V \\ \dim(V) = v}}\sum_{\langle e_j :j \in A \rangle \le W \le V} (-1)^{v-w} q^{\binom{v-w}{2}}\qbin{w}{u}{q}.
\end{multline*}
     Moreover we have
\begin{multline*}
    |\{V \le \F_q^k :e_i \notin V, \dim(V) = v, \langle e_j :j \in A \rangle \le V\}| =  \\ |\{V \le \F_q^k : \dim(V) = v, \langle e_j :j \in A \rangle \le V\}| - |\{V \le \F_q^k : \dim(V) = v, \langle e_j :j \in A \cup \{i\}\rangle \le V\}| = \\
    \qbin{k-a}{v-a}{q}-\qbin{k-a-1}{v-a-1}{q}.
\end{multline*}
From all of the above we get:
\begin{multline*}
    |\{M \in \F_q^{k \times (s-a)} :e_i \notin \mbox{colsp}(M) + \langle e_j :j \in A \rangle\}| = \\
    \sum_{u=0}^{k-1} |\{U \le \F_q^k :\dim(U)=u, e_i \notin U+\langle e_j :j \in A \rangle\}|\phi(u,s-a) = \\
    \sum_{u=0}^{k-1} \sum_{v=0}^{k-1}\sum_{\substack{V \le \F_q^k \\ e_i \notin V \\ \dim(V) = v}}\sum_{\langle e_j :j \in A \rangle \le W \le V} (-1)^{v-w} q^{\binom{v-w}{2}}\qbin{w}{u}{q}\phi(u,s-a) = \\
\sum_{u=0}^{k-1} \sum_{v=0}^{k-1}\left(  \qbin{k-a}{v-a}{q}-\qbin{k-a-1}{v-a-1}{q}\right)\sum_{\langle e_j :j \in A \rangle \le W \le V} (-1)^{v-w} q^{\binom{v-w}{2}}\qbin{w}{u}{q}\phi(u,s-a) = \\
\sum_{u=0}^{k-1} \sum_{v=0}^{k-1}\left(  \qbin{k-a}{v-a}{q}-\qbin{k-a-1}{v-a-1}{q}\right)\sum_{w=a}^v \qbin{v-a}{w-a}{q} (-1)^{v-w} q^{\binom{v-w}{2}}\qbin{w}{u}{q}\phi(u,s-a).
\end{multline*}
Our final average is
\begin{multline*}
      \sum_{s=0}^{n-1} \frac{1}{q^{k(n-k)}\binom{n-1}{s}} \sum_{a=0}^{k-1} \binom{k-1}{a} 
\binom{n-k}{s-a} q^{k(n-k-s+a)}\cdot \\\sum_{u=0}^{k-1} \sum_{v=0}^{k-1}\left(  \qbin{k-a}{v-a}{q}-\qbin{k-a-1}{v-a-1}{q}\right)\sum_{w=a}^v \qbin{v-a}{w-a}{q} (-1)^{v-w} q^{\binom{v-w}{2}}\qbin{v}{u}{q}\phi(u,s-a).
\end{multline*}
\end{proof}
\phantom\qedhere 
\end{proof}

The following three lemmas are used in Subsection~\ref{sec:extsim} as tools for proving the statement in Theorem~\ref{thm:extsim}. We use the notation introduced in Subsection~\ref{sec:extsim}.

\begin{lemma} \label{lemm1}
    Let $\eta(z,\omega)$ denote the number of subspaces $V \le \F_q^k$ with $\dim(V)=z$
    and $\omega(V)=\omega$. We have
    $$\eta(z,\omega)=\binom{k}{\omega}\sum_{r=\omega}^k (-1)^{r-\omega} \binom{k-\omega}{r-\omega}\qbin{k-r}{z-r}{q}.$$
\end{lemma}
\begin{proof}
For a subset $L \subseteq \{0, \ldots k\}$, let
$a(L)=|\{V \le \F_q^k \mid \dim(V)=z, \, I(V)=L\}|$ and 
$b(L)=\sum_{R \supseteq L} a(R)$.
Note that $$b(L)=|\{V \le \F_q^k \mid \dim(V)=z, \, e_t \in V \mbox{ for all } t \in L\}|=\qbin{k-|L|}{z-|L|}{q}$$
for all $L$. Therefore by Möbius Inversion in the Boolean algebra over the set $\{1, \ldots, k\}$ we have 
$$a(L)=\sum_{R \supseteq L} b(R)(-1)^{|R|-|L|}
=\sum_{r=|L|}^k \binom{k-|L|}{r-|L|}\qbin{k-r}{z-r}{q}(-1)^{r-|L|}$$
for all $L$.
Then, by definition,
$$\eta(z,\omega)= \binom{k}{\omega}\sum_{r=\omega}^k \binom{k-\omega}{r-\omega}\qbin{k-r}{z-r}{q}(-1)^{r-\omega},$$
as the lemma claims.
\end{proof}

\begin{lemma}\label{lemm2}
    Let $\eta_i(z,\omega)$ denote the number of subspaces $V \le \F_q^k$ with $\dim(V)=z$,  $\omega(V)=\omega$, and $e_i \in V$. We have    \begin{align*}\eta_i(z,\omega)&= \binom{k-1}{\omega-1}\sum_{r=\omega}^k \binom{k-\omega}{r-\omega}\qbin{k-r}{z-r}{q}(-1)^{r-\omega} \\
&
+\binom{k-1}{\omega}\sum_{r=\omega+1}^k \binom{k-\omega-1}{r-\omega-1}\qbin{k-r}{z-r}{q}(-1)^{r-\omega} \\
&+ 
\binom{k-1}{\omega}\sum_{r=\omega}^{k-1} \binom{k-1-\omega}{r-\omega}\qbin{k-r-1}{z-r-1}{q}(-1)^{r-\omega}.
\end{align*}
Note that $\eta_i(z,\omega)$ does not depend on $i$.
\end{lemma}
\begin{proof}
We argue as in the proof of Lemma~\ref{lemm1}, but imposing the extra condition that $e_i \in V$. For $L \subseteq \{0, \ldots, k\}$, let
$a(L)=|\{V \le \F_q^k \mid \dim(V)=z, \, I(V)=L, \, e_i \in V\}|$ and 
$b(L)=\sum_{R \supseteq L} a(R)$.
Then, by definition, $$b(L)=|\{V \le \F_q^k : \dim(V)=z, \, e_t \in V \mbox{ for all } t \in L, \, e_i \in V\}|.$$
Note that
$$b(L) = \begin{cases}
    \qbin{k-|L|}{z-|L|}{q} & \mbox{if $i \in L$,} \\ 
    \qbin{k-|L|-1}{z-|L|-1}{q} & \mbox{if $i \notin L$.}    
\end{cases}$$
Therefore, if $i \in L$ by Möbius Inversion in the Boolean algebra over the set $\{0, \ldots, k\}$ we  have
$$a(L)=\sum_{R \supseteq L} b(R)(-1)^{|R|-|L|}
=\sum_{r=|L|}^k \binom{k-|L|}{r-|L|}\qbin{k-r}{z-r}{q}(-1)^{r-|L|}.$$
If $i \notin L$, again by Möbius Inversion we have
\begin{align*}
a(L) &= \sum_{\substack{R \supseteq L \\ R \ni i}}
b(R) (-1)^{|R|-|L|} + \sum_{\substack{R \supseteq L \\ R \not\ni i}}
b(R) (-1)^{|R|-|L|} \\
&= \sum_{r=|L|+1}^k \binom{k-|L|-1}{r-|L|-1}\qbin{k-r}{z-r}{q}(-1)^{r-|L|}
+
\sum_{r=|L|}^{k-1} \binom{k-1-|L|}{r-|L|}\qbin{k-r-1}{z-r-1}{q}(-1)^{r-|L|}.
\end{align*}
Finally, by definition,
\begin{align*}
\eta_i(z,\omega)&= \binom{k-1}{\omega-1}\sum_{r=\omega}^k \binom{k-\omega}{r-\omega}\qbin{k-r}{z-r}{q}(-1)^{r-\omega} \\
&
+\binom{k-1}{\omega}\sum_{r=\omega+1}^k \binom{k-\omega-1}{r-\omega-1}\qbin{k-r}{z-r}{q}(-1)^{r-\omega} \\
&+ 
\binom{k-1}{\omega}\sum_{r=\omega}^{k-1} \binom{k-1-\omega}{r-\omega}\qbin{k-r-1}{z-r-1}{q}(-1)^{r-\omega},
\end{align*}
as claimed.
\end{proof}

In the proof of Theorem~\ref{thm:extsim} we will also need the following preliminary observation, for which the proof is omitted.
\begin{lemma}\label{lemm3}
For a subspace $U \le \F_q^k$ of dimension $z$, 
let $v_i(h,U)$ be the number of subspaces $V \le \F_q^k$ of dimension $h$ with $e_i \in V$ and $V \supseteq U$.
We have
$$v_i(h,U)= \begin{cases} \displaystyle
     \qbin{k-z}{h-z}{q} & \mbox{if $e_i \in U$,} \\ \displaystyle
      \qbin{k-z-1}{h-z-1}{q} & \mbox{if $e_i \notin U$.}
 \end{cases}$$
\end{lemma}

\end{appendix}
\end{document}